\documentclass[a4paper, 12 pt, twoside, reqno]{amsart}

\usepackage[english]{babel}
\usepackage[T1]{fontenc}
\usepackage[utf8]{inputenc}

\usepackage{setspace}
\usepackage[euler-digits,euler-hat-accent]{eulervm}

\usepackage{times}
\usepackage{amsmath}
\usepackage{amsfonts}
\usepackage{amssymb}
\usepackage{amsmath}
\usepackage{amsthm}
\usepackage{graphicx}
\usepackage{array}
\usepackage{color}
\usepackage{mathrsfs}
\usepackage{hyperref}
\usepackage{eucal}
\usepackage{esint}
\usepackage{tikz}
\usepackage{upgreek}
\usepackage{enumitem}

\allowdisplaybreaks

\theoremstyle{plain}
\newtheorem{theorem}{Theorem}[section]
\newtheorem{corollary}[theorem]{Corollary}
\newtheorem{proposition}[theorem]{Proposition}
\newtheorem{lemma}[theorem]{Lemma}

\theoremstyle{definition}
\newtheorem{definition}[theorem]{Definition}

\theoremstyle{remark}
\newtheorem{remark}[theorem]{Remark}

\numberwithin{equation}{section}
\numberwithin{figure}{section}
\numberwithin{table}{section}

\newcommand{\R}{\mathbb{R}}
\newcommand{\N}{\mathbb{N}}
\newcommand{\C}{\mathbb{C}}

\newcommand{\Z}{\mathbb{Z}}

\newcommand{\s}[1]{\CMcal{#1}}

\newcommand{\bb}[1]{\mathscr{#1}}
\newcommand{\rr}[1]{\mathfrak{#1}}

\newcommand{\ketbra}[2]{|#1\rangle\langle#2|}

\newcommand{\expo}[1]{\,\mathrm{e}^{#1}\,}

\newcommand{\dd}{\,\mathrm{d}}
\newcommand{ \ii}{\,\mathrm{i}\,}

\newcommand{\virg}[1]{\lq\lq#1\rq\rq}                \newcommand{\ie}{\textsl{i.\,e.\,}}
\newcommand{\eg}{\textsl{e.\,g.\,}}
\newcommand{\cf}{\textsl{cf}.\,}

\DeclareMathOperator{\Tr}{Tr}

\newcommand{\HC}{H\!C}

\hypersetup{
	pdftoolbar=true,
	pdfmenubar=true,
	pdffitwindow=true,
	pdfstartview=true,
	pdftitle={NCG-Landau-II},
	pdfauthor={Giuseppe De Nittis},
	breaklinks=true,
	colorlinks=true,
	linkcolor=purple,
	citecolor=teal,
	urlcolor=blue,
	bookmarksopen=true,
	filecolor=magenta,
}

\begin{document}

\title[The noncommutative geometry of the Landau Hamiltonian]
{The noncommutative geometry of the Landau Hamiltonian: Differential aspects}

\author[G. De~Nittis]{Giuseppe De Nittis}
\address[G. De~Nittis]{Facultad de Matemáticas \& Instituto de Física,
	Pontificia Universidad Católica de Chile,
	Santiago, Chile.}
\email{gidenittis@mat.uc.cl}

\author[M. Sandoval]{Maximiliano Sandoval}
\address[M. Sandoval]{Facultad de Matemáticas,
	Pontificia Universidad Católica de Chile,
	Santiago, Chile.}
\email{msandova@protonmail.com}

\vspace{2mm}

\date{\today}

\maketitle

\begin{abstract}
	{In this work we study the differential aspects of the noncommutative geometry for the magnetic $C^*$-algebra
		which is a 2-cocycle deformation of the group $C^*$-algebra of $\R^2$.
		This algebra is intimately related to the study of the Quantum Hall Effect in the continuous,
		and our results aim to provide a new geometric interpretation of the related Kubo's formula. Taking inspiration from the ideas
		developed by Bellissard during the 80's, we build an appropriate  Fredholm module for the
		magnetic $C^*$-algebra based on the magnetic Dirac operator
		which is the square root (\`a la Dirac) of the quantum harmonic oscillator. Our main result consist of
		establishing an important piece of Bellissard's theory, the so-called second Connes'
		formula. In order to do so, we establish the  equality of three cyclic 2-cocycles defined on
		a dense subalgebra of the   magnetic $C^*$-algebra. Two of these 2-cocycles are new in the literature and are defined by Connes' quantized
		differential calculus, with the use of the Dixmier trace and  the magnetic Dirac operator. }

	\medskip


\end{abstract}

\tableofcontents

\section{Introduction}\label{sec:BG_mat}

{This work continues the  study of the noncommutative geometry of the magnetic $C^*$-algebra $\bb{C}_B$ associated with the Landau Hamiltonian started in~\cite{denittis-sandoval-00}. While the previous work has been devoted to the analysis of metric aspects, in the present work we will investigate the topological properties by developing an appropriate quantized calculus based on the spectral triple introduced in~\cite{denittis-sandoval-00} and endowed with a suitable grading. The main result of this paper is the proof of the equality of three cyclic 2-cocycles $\Psi_B$, $\rr{Ch}_B$  and  $\tau_{B,2}$ defined on
	a dense subalgebra of $\bb{C}_B$. The  2-cocycle $\Psi_B$ is  standard in the literature concerning the topology of $\bb{C}_B$
	while $\rr{Ch}_B$  and  $\tau_{B,2}$ are new and are defined by Connes' quantized
	differential calculus, with the use of the Dixmier trace and the the spectral triple introduced in~\cite{denittis-sandoval-00}. The equalities $\Psi_B=\rr{Ch}_B$ and $\Psi_B=\tau_{B,2}$, called the \emph{second Connes' formulae}
	in agreement with the name used in the seminal paper \cite{bellissard-elst-schulz-baldes-94}, provide
	a new way of representing the Kubo’s formula   for the Quantum Hall effect inside the noncommutative geometry of the magnetic $C^*$-algebra $\bb{C}_B$.
	In particular, the construction of $\tau_{B,2}$ requires the  introduction of the notion of \emph{quasi-even Fredholm} which can be considered as a new
	idea in noncommutative geometry extending the usual concept
	of Fredholm module. Our hope is that this idea could be of some interest also  for further applications in noncommutative geometry. In the rest of this introduction we will give a more detailed account of our results by comparing them with the existing literature.}

\subsection{Background material and known results}
In order to describe the main results of this work, we will first proceed to
introduce the necessary background. The  material and the notation presented
below  are borrowed from~\cite{denittis-gomi-moscolari-19, denittis-sandoval-00}.

\medskip

Consider  the Hilbert space $L^2(\R^2)$, and let $\{\psi_{n,m}\}\subset L^2(\R^2)$, with
$n,m\in\N_0:=\N\cup\{0\}$, be the orthonormal basis provided by the
\emph{generalized Laguerre
	basis} defined by
\begin{equation}\label{eq:lag_pol}
	\psi_{n,m}(x)\;:=\;\psi_{0,0}(x)\ \sqrt{\frac{n!}{m!}}\left[\frac{x_1+\ii x_2}{\sqrt{2}\ell_B}\right]^{m-n}L_{n}^{(m-n)}\left(\frac{|x|^2}{2\ell_B^2}\right)\; ,
\end{equation}
where
\[
	L_n^{(\alpha)}\left(\zeta\right)\;:=\;\sum_{j=0}^{n}\frac{(\alpha+n)(\alpha+n-1)\ldots(\alpha+j+1)}{j!(n-j)!}\left(-\zeta\right)^j\;,\qquad \alpha,\zeta\in \R
\]
are the {generalized Laguerre polynomial} of degree $m$ (with the usual convention $0!=1$) and
\begin{equation}\label{eq:herm1}
	\psi_{0,0}(x)\;:=\;\frac{1}{\sqrt{2\pi}\ell_B}\ \expo{-\frac{|x|^2}{4\ell_B^2}}\;.
\end{equation}
The parameter $\ell_B>0$ is called
\emph{magnetic length} and the (singular) limit $\ell_B\to +\infty$ corresponds to the
limit where the magnetic field $B$ vanishes.
Let us introduce the family $\left\{\Upsilon_{j\mapsto k}\;|\, (j,k)\in \N_0^2\right\}$ of \emph{transition operators} on $L^2(\R^2)$  defined by
\begin{equation}\label{eq:intro:basic_op}
	\Upsilon_{j\mapsto k}\psi_{n,m}\;:=\;\delta_{j,n}\;\psi_{k,m}\;,\qquad k,j,n,m\in\N_0\;.
\end{equation}
A direct computation~\cite[Proposition 2.10]{denittis-sandoval-00} provides
\begin{equation}\label{eq:intro_01}
	(\Upsilon_{j\mapsto k})^*\;=\;\Upsilon_{k\mapsto j}\;,\qquad
	\Upsilon_{j\mapsto k}\Upsilon_{m\mapsto n}=\delta_{j,n}\Upsilon_{m\mapsto k}\;,
\end{equation}
and with these rules in hand one can define the \emph{magnetic $C^*$-algebra}
\begin{equation}\label{eq:equal_C-Ups}
	\bb{C}_B\;=\;C^*(\Upsilon_{j\mapsto k},\;k,j\in\N_0)
\end{equation}
as the $C^*$-algebra generated by the {transition} operators. This name is justified by the fact that the \emph{Landau projections}
\[
	\Pi_j\;:=\; \Upsilon_{k\mapsto j}\Upsilon_{j\mapsto k}\;=\;\sum_{r\in\N_0}\ketbra{\psi_{j,r}}{\psi_{j,r}}\;,\qquad j\in\N_0
\]
(independently of $k$) are elements of $\bb{C}_B$, the latter being the spectral projection of the Landau Hamiltonian
\begin{equation}\label{eq:int_LH}
	H_{B}= \frac{\epsilon_B}{2} \left(K_{1}^{2} + K_{2}^{2}\right),
\end{equation}
where
\begin{equation}\label{eq:momenta}
	K_{1} \;=\; -\ii \ell_{B}\frac{\partial}{\partial x_{1}}  - \frac{1}{2 \ell_{B}} x_{2}\;,\qquad
	K_{2} \;=\; -\ii \ell_{B}\frac{\partial}{\partial x_{2}}  + \frac{1}{2 \ell_{B}} x_{1}
\end{equation}
are
the \emph{magnetic momenta} and the constant $\epsilon_B$ is the fundamental \emph{magnetic energy}.

\medskip

There are interesting spaces of operators contained in $\bb{C}_B$. Let us introduce the following notation
\begin{equation}\label{eq:exp_op}
	\begin{aligned}
		\bb{S}_B\;   & :=\;\left.\left\{A\;:=\;\sum_{(j,k)\in\N_0^2}a_{j,k}\Upsilon_{j\mapsto k}\;\right|\;\{a_{j,k}\}\in S(\N_0^2)\right\}\;,      \\
		\bb{L}^p_B\; & :=\;\left.\left\{A\;:=\;\sum_{(j,k)\in\N_0^2}a_{j,k}\Upsilon_{j\mapsto k}\;\right|\;\{a_{j,k}\}\in \ell^p(\N_0^2)\right\}\;,
	\end{aligned}
\end{equation}
where
$S(\N_0^2)$ is the space
of \emph{rapidly decreasing} sequences,
and $\ell^p(\N_0^2)$ are the usual discrete $L^p$ spaces. It turns out that~\cite[Proposition 2.17]{denittis-sandoval-00}
$$
	\bb{S}_B\;\subset\;\bb{L}^1_B\;\subset\;\bb{I}_B\;\subset\;\bb{L}^2_B\;\subset\;\bb{C}_B\;\subset\;\bb{M}_B\;,
$$
where
$$
	\bb{I}_B\;:=\;\left\{S=AB\;|\; A,B\in\bb{L}_{B}^2 \right\}\;\equiv\;\left(\bb{L}_{B}^2\right)^2\;
$$
and $\bb{M}_B$ is the enveloping von Neumann algebra of $\bb{C}_B$.
All these subspaces are dense in $\bb{C}_B$ with respect to the operator norm,
and in $\bb{M}_B$ with respect to the weak or strong operator topologies. Both $\bb{L}^2_B$, and consequently $\bb{I}_B$, are
self-adjoint two-sided ideals of $\bb{M}_B$.
The spaces $\bb{S}_B$ and $\bb{L}^2_B$ admit special characterizations
in terms of integral kernel operators. Let us start with $\bb{L}^2_B$ (\cf
\cite[Section 2.4]{denittis-sandoval-00}). One gets that $A\in \bb{L}^2_B$, if and
only if, there is a function $f_A\in L^2(\R^2)$ such that
\begin{equation}\label{eq:bg1}
	(A\varphi)(x)\;=\;\frac{1}{2\pi \ell_B^{2}} \int_{\mathbb{R}^{2}}\dd y\; f_A(y-x)\;\Phi_B(x,y)\;\varphi(y)\;,
	\;
	\quad \forall\;\varphi \in L^2(\R^2)\;
\end{equation}
where the function
$$
	\Phi_B(x,y)\;:=\;\expo{\ii\frac{x_{1}y_{2} -  x_{2}y_{1}}{2\ell_{B}^{2}}}\;,\qquad x,y\in\R^2
$$
is known as \emph{magnetic  2-cocycle}.
The relation between the integral kernel $f_A$ and the sequence
$\{a_{j,k}\}\in \ell^2(\N_0^2)$ which identifies the expansion of $A$ in the basis $\Upsilon_{j\mapsto k}$ is given by
\begin{equation}\label{eq:bg2}
	f_A\;=\;\sqrt{2\pi}\,\ell_B\sum_{(j,k)\in\N_0^2}(-1)^{j-k}a_{j,k}\;\psi_{k,j}\;
\end{equation}
and the norm bound $\sqrt{2\pi}\ell_B\|A\|\leqslant\|f_A\|_{L^2}$ holds true.
A similar result holds for $\bb{S}_B$, namely  $A\in \bb{S}_B$, if and only if, there is a Schwarz function $f_A\in S(\R^2)$ such that
$A$ has an integral representation of the type~\eqref{eq:bg1}
and the relation between $A$ and its kernel is given again by \eqref{eq:bg2}. In addition, $\bb{S}_B$ has the structure of  a
Fréchet pre-$C^*$-algebra of $\bb{C}_B$~\cite[Proposition 2.8 \& Proposition 2.14]{denittis-sandoval-00}. Behind the integral representation \eqref{eq:bg1} there is the fact that $\bb{C}_B$ is nothing more than the \emph{group $C^*$-algebra} of $\R^2$ \emph{twisted} by the cocycle  $\Phi_B$ (\cf~\cite[Section 2.2]{denittis-sandoval-00} and references therein).

\medskip

As discussed in~\cite[Section 2.6]{denittis-sandoval-00}, one can  endow
the von Neumann algebra $\bb{M}_B$ with a remarkable \emph{normal}, \emph{faithful} and \emph{semi-finite} (NFS) trace  $\fint_B$ defined
on the ideal $\bb{I}_B$, which is uniquely specified by the prescription
\begin{equation}\label{eq:trac_prod}
	{\fint_B}(A^*B)\;:=\frac{1}{2\pi\ell_B^2}\langle f_A,f_B\rangle_{L^2}\;,\qquad \forall\; A,B\in\bb{L}^2_B
\end{equation}
where $\langle\;,\;\rangle_{L^2}$ is the usual inner product in $L^2(\R^2)$ and $f_A,f_B\in L^2(\R^2)$ are the integral kernels of $A$ and $B$ respectively, as given by the prescription \eqref{eq:bg2}. The computation of the trace $\fint_B$ on elements of the domain $\bb{I}_B$ is facilitated by observing that every $S\in \bb{I}_B$ has an integral kernel
of type \eqref{eq:bg2} which satisfies $f_S\in L^2(\R^2)\cap C_0(\R^2)$, where
$C_0(\R^2)$ is the space of continuous functions which vanish at infinity. On
these elements the trace can be computed as $\fint_B(S)= f_S(0)$~\cite[Corollary 2.22]{denittis-sandoval-00}. The trace $\fint_B$ has the physical meaning of a thermodynamic limit. Indeed, one can prove that~\cite[Lemma 2.23]{denittis-sandoval-00}
\begin{equation}\label{eq:intr_10}
	\fint_B(S)\;=\;2\pi\ell_B^2\;\lim_{n\to+\infty} \frac{1}{|\Lambda_n|}{\Tr}_{L^2(\R^2)}( \chi_{\Lambda_n}S \chi_{\Lambda_n} )\;,\qquad S\in \bb{I}_B
\end{equation}
where the family $\{\Lambda_n\}$ provides an
increasing sequence of compact subsets $\Lambda_n \subseteq \R^2$ such that $\Lambda_n\nearrow\R^2$ and which satisfies the \emph{F{\o}lner condition} (see \eg~\cite{greenleaf-69} for more details),
$|\Lambda_n|$ is the Lebesgue measure of $\Lambda_n$ and $\chi_{\Lambda_n}$ is the projection defined as the multiplication operator by the characteristic function of $\Lambda_n$.
The expression on the right-hand side of \eqref{eq:intr_10} is known as \emph{trace per unit of volume}.

\medskip

The magnetic algebra $\bb{C}_B$ admits a pair of unbounded \emph{spatial derivations} which can be initially defined on the pre-$C^*$-algebra $\bb{S}_B$ by the commutators
\begin{equation}\label{eq:commut_deriv_01_intr}
	\nabla_j A\;:=\;-\ii [x_j,A],\qquad j=1,2,\quad A\in \bb{S}_B,
\end{equation}
where $x_j$ are the position operators on $L^2(\R^2)$. By closing with respect to suitable Fréchet-type norms one can define the Banach spaces $C^N(\bb{C}_B)$ of $N$-times differentiable elements (\cf~\cite[Section 2.8]{denittis-sandoval-00}). Remarkably, one has that $\bb{S}_B\subset C^\infty(\bb{C}_B)$  is made by  \emph{smooth elements}, namely by elements  which can be derived an indefinite number of times.

\medskip

The $K$-theory of  $\bb{C}_B$ is quite simple to compute. From
\cite[Proposition 2.11]{denittis-sandoval-00} we know that there is an isomorphism of $C^*$-algebras $\bb{C}_B\simeq \bb{K}$ where $\bb{K}$ is the $C^*$-algebra of compact operators. Since the $K$-theory is invariant under $C^*$-isomorphisms one immediately gets
$K_0(\bb{C}_B)\simeq\Z$ and $K_1(\bb{C}_B)=0$. A more precise description of the $K_0$-group is given by.
$$
	K_0(\bb{C}_B)\;\simeq\;K_0(\bb{S}_B)\;=\;\Z[\Pi_0]\;.
$$
The first isomorphism is justified by the fact that $\bb{S}_B$ is a pre-$C^*$-algebra of $\bb{C}_B$~\cite[Theorem 3.44]{gracia-varilly-figueroa-01} and the last equality follows by an inspection of the isomorphism $\bb{C}_B\simeq \bb{K}$. It is worth noting that since $\Pi_0\in \bb{S}_B$ then the $K$-theory of   $\bb{C}_B$ is realized inside $\bb{S}_B$. Moreover, since all the Landau projections are equivalent (in the sense of von Neumann)~\cite[Lemma 5]{bellissard-elst-schulz-baldes-94}
one has that $[\Pi_0]=[\Pi_j]$ for every $j\in\N_0$.

\medskip

The trace $\fint_B$ is a cyclic $0$-cocycle of the algebra $\bb{S}_B$ and so it defines a class $[\fint_B]\in HC^{\rm even}(\bb{S}_B)$
in the even cyclic cohomology of $\bb{S}_B$ (see Appendix \ref{app:cyc_cohom}).
Given the canonical pairing $\langle\;,\;\rangle:HC^{\rm even}(\bb{S}_B)\times K_0(\bb{S}_B)\to\C$
between the even cyclic cohomology and the even $K$-theory one can define the map
$$
	gl_B([P])\;:=\;\langle{\textstyle [\fint_B]},[P]\rangle\;=\;\fint_B(P)\;,\qquad [P]\in K_0(\bb{S}_B)
$$
where   $P\in \bb{S}_B$ is any representative of the class $[P]$
in view of the fact that the $K$-theory is entirely realized inside the algebra.
The map $gl_B$ is known as the \emph{gap labeling function}~\cite{bellissard-86,bellissard-93}, and in our specific case, it provides the group isomorphism
\begin{equation}\label{eq:BM_01}
	gl_B\;:\;K_0(\bb{S}_B)\;\stackrel{\simeq}{\longrightarrow}\;\Z
\end{equation}
generated by $\fint_B(\Pi_0)=1$~\cite[eq. (2.22)]{denittis-sandoval-00}. It is
worth mentioning that the last result is a special case of~\cite[Theorem
	2.2]{xia-88} when the hull of the potentials collapses to a singleton due to the circumstance that we are considering no
electrostatic interactions.

\medskip

By combining the trace   $\fint_B$ and the derivations $\nabla_j$ one gets the cyclic $2$-cocycle $\Psi_B$ defined by
\begin{equation}\label{eq:2-cocy-Psi}
	\Psi_B(A_0,A_1,A_2)\;:=\;\fint_B\big(A_0(\nabla_1A_1\nabla_2A_2-\nabla_2A_1\nabla_1A_2)\big)\;,
\end{equation}
for every $A_0,A_1,A_2\in \bb{S}_B$.
This provides a second class $[\Psi_B]\in HC^{\rm even}(\bb{S}_B)$ and a second formula for  the canonical pairing with the $K$-theory defined by
\begin{equation}\label{eq:BM_020}
	c_{B}([P])\;:=\;\frac{\ii}{\ell^2_B}\langle[\Psi_B],[P]\rangle\;=\;\frac{\ii}{\ell^2_B}\Psi_B(P,P,P)\;,\qquad [P]\in K_0(\bb{S}_B)
\end{equation}
where $P\in \bb{S}_B$ is any representative of  $[P]$ inside the algebra. The map $c_{B}$ provides the \emph{Chern number} of the class $[P]$ (or of the projection $P$ with a little abuse of terminology) and defines a second
group isomorphism
\begin{equation}\label{eq:BM_02}
	c_{B}\;:\;K_0(\bb{S}_B)\;\stackrel{\simeq}{\longrightarrow}\;\Z
\end{equation}
generated by $c_B(\Pi_0)=1$~\cite[Section 3.7]{denittis-gomi-moscolari-19}. Again, the integrality of the map $c_B$
above can be seen as a special case of~\cite[Theorem 3.3]{xia-88} when the hull of the potentials collapses to a singleton. In view of $HC^{\rm even}(\bb{S}_B)\simeq \Z$ (Lemma \ref{lemma:comp_cohom}) one infers that $[\fint_B]=[\Psi_B]$
and therefore one has the equality
\begin{equation}\label{eq:str_for}
	gl_B([P])\;=\;c_{B}([P])\;,\qquad\forall\; [P]\in K_0(\bb{S}_B)\;.
\end{equation}

\medskip

The maps \eqref{eq:BM_01} and \eqref{eq:BM_02} have  important physical manings in the context of the geometric interpretation of the Quantum Hall Effect~\cite{bellissard-86,xia-88,bellissard-elst-schulz-baldes-94}.
Let $H$ be a possibly unbounded self-adjoint operator affiliated to $\bb{M}_B$. Assume that the spectrum of $H$ is bounded from below,
and for  every (Fermi) energy $E\in\rho(H)$ in the resolvent set of $H$, the spectral projection $P_E:=\chi_{(-\infty,E)}(H)$ lies in the pre-$C^*$ algebra $\bb{S}_B$. In this case
\begin{equation}\label{eq:int_idos}
	N_H(E)\;:=\; \frac{1}{2\pi\ell^2_B}\;gl_B([P_E])
\end{equation}
provides the \emph{integrated density of states} of $H$ inside the spectral gap detected by $E$~\cite{veselic-08} and
\begin{equation}\label{eq:int_hall_cond}
	\sigma_H(E)\;:=\;\frac{e^2}{2\pi \hslash}\;c_B([P_E])
\end{equation}
is the \emph{Hall conductance} associated to  the energy spectrum of $H$ below the (Fermi) energy $E$ (the prefactor has the physical units of a conductance).
For instance, the results above apply to the Landau Hamiltonian $H_B$ given by \eqref{eq:int_LH} since the Landau projections $\Pi_j$ are in  $\bb{S}_B$. In this context the equality \eqref{eq:str_for} is known as \emph{Str\v{e}da formula}.

\subsection{New results}\label{int_new}
The  main novelty of this work is to reformulate the results presented in the previous section, and in particular the integrality of the maps \eqref{eq:BM_01}  and \eqref{eq:BM_02}, in the context of the geometry of the \emph{magnetic spectral triple} $(\bb{S}_B,\s{H}_4, D_B)$ introduced in~\cite{denittis-sandoval-00}. The latter is defined by the Hilbert space
\begin{equation}\label{eq:H4}
	\s{H}_4\;:=\;L^2(\R^2)\;\otimes\;\C^4,
\end{equation}
on which the von Neumann algebra $\bb{M}_B$, along with each of its subalgebras like $\bb{S}_B$, are represented diagonally, \ie
$$
	\pi\;:\;A\;\longmapsto\;A\otimes{\bf 1}_4\;=\;\left(\begin{array}{c c c c}
			A & 0 & 0 & 0 \\
			0 & A & 0 & 0 \\
			0 & 0 & A & 0 \\
			0 & 0 & 0 & A
		\end{array}
	\right)\;,\qquad A\in  \bb{M}_B\;.
$$
The \emph{magnetic Dirac operator} is defined by
\begin{equation}\label{eq:intro_D}
	D_B\;:=\;\frac{1}{\sqrt{2}}\big(K_1\;\otimes\;\gamma_1\;+\;K_2\;\otimes\;\gamma_2\;+\;G_1\;\otimes\;\gamma_3\;+\;G_2\;\otimes\;\gamma_4\big)
\end{equation}
where $K_1$ and $K_2$ are the  {magnetic momenta} \eqref{eq:momenta}, and
$G_1$ and $G_2$ are the \emph{dual magnetic momenta} given by
\begin{equation}\label{eq:dual-momenta}
	G_{1}\; =\; -\ii \ell_{B}\frac{\partial}{\partial x_{2}}  - \frac{1}{2 \ell_{B}} x_{1}\;,\qquad
	G_{2}\; =\; -\ii \ell_{B}\frac{\partial}{\partial x_{1}}  + \frac{1}{2 \ell_{B}} x_{2}
\end{equation}
and $\gamma_1,\ldots,\gamma_4$ is any set of Hermitian $4\times 4$ matrices which satisfy the fundamental anti-commutation relations of the
Clifford algebra $C\ell_4(\C)$. Without loss of generality  will fix the following convenient choice\footnote{It is worth noting that the definition of the $\gamma$-matrices differs from that in \cite[p. 31]{denittis-sandoval-00}. However the two set of $\gamma$-matrices  are related by the unitary involution
	$$
		I\; :=\; \left(
		\begin{array}{c c c c}
				0 & 0 & 0 & 1 \\
				0 & 1 & 0 & 0 \\
				0 & 0 & 1 & 0 \\
				1 & 0 & 0 & 0
			\end{array}
		\right)\;.
	$$}:
\begin{equation}\label{eq:gamm_mat}
	\begin{aligned}
		\gamma_{1}\; & :=\; \left(
		\begin{array}{c c c c}
				0 & 0 & 0 & 1 \\
				0 & 0 & 1 & 0 \\
				0 & 1 & 0 & 0 \\
				1 & 0 & 0 & 0
			\end{array}
		\right)\;,\qquad
		             &              & \gamma_{2}\;: =\; \left(\begin{array}{c c c c}
				0    & 0    & 0   & \ii \\
				0    & 0    & \ii & 0   \\
				0    & -\ii & 0   & 0   \\
				-\ii & 0    & 0   & 0
			\end{array}
		\right)\;,
		\\
		\gamma_{3}\; & : =\; \left(
		\begin{array}{c c c c}
				0  & -1 & 0 & 0 \\
				-1 & 0  & 0 & 0 \\
				0  & 0  & 0 & 1 \\
				0  & 0  & 1 & 0
			\end{array}
		\right)\;,\qquad
		             &              & \gamma_{4}\;: =\; \left(\begin{array}{c c c c}
				0   & -\ii & 0   & 0    \\
				\ii & 0    & 0   & 0    \\
				0   & 0    & 0   & -\ii \\
				0   & 0    & \ii & 0
			\end{array}
		\right)\;.
	\end{aligned}
\end{equation}
The {magnetic Dirac operator} is essentially self-adjoint on the dense domain $S(\R^2)\otimes\C^4$ and has compact resolvent~\cite[Proposition 3.1]{denittis-sandoval-00}.
By a straightforward computation one gets
\[
	D_B^2\;:=\;\left(\begin{array}{c c c c}
			Q_B & 0   & 0   & 0   \\
			0   & Q_B & 0   & 0   \\
			0   & 0   & Q_B & 0   \\
			0   & 0   & 0   & Q_B
		\end{array}
	\right)\;+\;\left(\begin{array}{c c c c}
			-{\bf 1} & 0 & 0       & 0 \\
			0        & 0 & 0       & 0 \\
			0        & 0 & {\bf 1} & 0 \\
			0        & 0 & 0       & 0
		\end{array}
	\right)\;
\]
where the operator
\begin{equation}\label{eq:harm_osc}
	Q_{B}\;: =\;  \frac{1}{2}\left(K_{1}^{2} + K_{2}^{2} + G_{1}^{2} + G_{2}^{2}\right)
\end{equation}
is the
two-dimensional isotropic \emph{harmonic oscillator} on $L^2(\R^2)$. The latter is diagonalized on the
Laguerre
basis
according to
\[
	Q_B \psi_{n,m}\;=\;(n+m+1)\ \psi_{n,m},\qquad\quad\;  (n,m)\in\N^2_0\;.
\]
As a consequence $Q_B$ has a pure point positive spectrum
with eigenvalues $\lambda_j:=j+1$, $j\in\N_0$, of finite {multiplicity} $\text{Mult}[\lambda_j]=j+1$.
The operator $D_B^2$ has a simple zero eigenvalue and therefore it is not invertible. For this reason  we need to introduce the regularized inverse powers
\begin{equation}\label{eq:reg-inv}
	{|D_{B,\varepsilon}|^{-s}}\;:=\;\left(D_B^2+\varepsilon{\bf 1}\right)^{-\frac{s}{2}}\;,\qquad \varepsilon>0\;,\quad s\geqslant 1\;.
\end{equation}

\medskip

The last ingredient we need to describe our first result is the \emph{Dixmier trace} $\Tr_{\rm Dix}$. There are several standard references
for the theory of the Dixmier trace, like~\cite[Chap.~4, Sect.~2]{connes-94},
~\cite[Appendix A]{connes-moscovici-95}, \cite[Sect.~7.5 and
	App.~7.C]{gracia-varilly-figueroa-01}, \cite{lord-sukochev-zanin-12},
\cite{alberti-matthes-02}, and we will refer to these sources  for the
construction and the properties of the Dixmier trace. A brief summary of the
most relevant information can be found in~\cite[Appendix
	B]{denittis-gomi-moscolari-19}. Here, we will  fix just  few notations (see also Appendix \ref{app_w-lp-dual}). The domain of definition of the Dixmier trace, called the \emph{Dixmier ideal}, will be denoted with $\rr{S}^{1^+}$.
The ideal  $\rr{S}^{1^+}_{0}\subset\rr{S}^{1^+}$ is the closure of the finite-rank operators in the norm of $\rr{S}^{1^+}$ and every Dixmier trace vanishes on $\rr{S}^{1^+}_{0}$.
The closed subspace of \emph{measurable elements} (those for which the Dixmier trace does not depend on the choice of scale-invariant generalized limit) will be denoted with $\rr{S}^{1^+}_{\rm m}$. Clearly $\rr{S}^{1^+}_{0}\subset \rr{S}^{1^+}_{\rm m}$. As proved in~\cite[Proposition 2.25]{denittis-sandoval-00},  one has that $|D_{B,\varepsilon}|^{-4}\in \rr{S}^{1^+}_{\rm m}$ and $\Tr_{\rm Dix}(|D_{B,\varepsilon}|^{-4})=2$.
However, this integrability property changes considerably when the quantity \eqref{eq:reg-inv} is \virg{dressed} with suitable elements of the magnetic $C^*$-algebra. Indeed from
~\cite[Proposition 2.27]{denittis-sandoval-00} one obtains that
\begin{equation}\label{eq:meas_prop_trip}
	|D_{B,\varepsilon}|^{-2}\;\pi(A)\;\in\; \rr{S}^{1^+}_{\rm m}\;,\qquad \forall\; A\in\bb{L}^1_B\;.
\end{equation}
Let us introduce the \emph{noncommutative integral} (a la Connes)
\begin{equation}\label{eq:meas_prop_trip_02}
	\rr{Int}_B(A)\;:=\;\frac{1}{4}\Tr_{\rm Dix}\left(|D_{B,\varepsilon}|^{-2}\;\pi(A)\right)\;.
\end{equation}
Then it holds true that~\cite[eq. (3.4)]{denittis-sandoval-00}
\begin{equation}\label{eq:meas_prop_trip_03}
	\rr{Int}_B(A)\;=\; \fint_B(A)\;,\qquad \forall\; A\in\bb{L}^1_B\;.
\end{equation}
Equalities \eqref{eq:meas_prop_trip_03} and~\eqref{eq:intr_10} also provide the proportionality constant between the {noncommutative integral} $\rr{Int}_B$ and the trace per unit of volume. Since $\bb{S}_B\subset \bb{L}^1_B$, one infers from \eqref{eq:meas_prop_trip} that the
	{magnetic spectral triple} $(\bb{S}_B,\s{H}_4, D_B)$ has \emph{spectral dimension} 2 as discussed in~\cite[Theorem 3.6]{denittis-sandoval-00}.

\medskip

Interestingly, the equality established by \eqref{eq:meas_prop_trip_02}, along with the properties of $\fint_B$, can be used to deduce that the {noncommutative integral}  $\rr{Int}_B$, as defined by \eqref{eq:meas_prop_trip_02}, is a
$0$-cocycle of the algebra $\bb{S}_B$ which provides a different representative for the class $[\fint_B]\in HC^{\rm even}(\bb{S}_B)$.
This fact provides a new way of computing the  gap labeling function in \eqref{eq:BM_01} via the noncommutative integral of the spectral triple $(\bb{S}_B,\s{H}_4, D_B)$.
\begin{theorem}[Gap labeling]\label{teo:01}
	Let $gl_B:K_0(\bb{S}_B)\to\Z$ be the gap labeling isomorphism defined in \eqref{eq:BM_01}. Then, it holds true that
	\begin{equation}
		\label{eq:teo-01}
		gl_B([P])\;=\;\rr{Int}_B(P)\;,\qquad [P]\in K_0(\bb{S}_B)\;
	\end{equation}
	where   $P\in \bb{S}_B$ is any representative of the class $[P]$.
\end{theorem}
For the description of the  second main result  we need the operator
\begin{equation}\label{eq:inv_Gamma}
	\Gamma\;:=\;{\bf 1}\otimes \ii\gamma_1\gamma_2\;=\;
	\left(\begin{array}{c c c c}
			{\bf 1} & 0       & 0        & 0        \\
			0       & {\bf 1} & 0        & 0        \\
			0       & 0       & -{\bf 1} & 0        \\
			0       & 0       & 0        & -{\bf 1}
		\end{array}
	\right)\;.
\end{equation}
This is a self-adjoint involution, \ie\ $\Gamma=\Gamma^*=\Gamma^{-1}$. By combining the Dirac operator $D_B$ and the involution $\Gamma$ one can define the expression
\begin{equation}\label{eq:2-cocy-Psi-Dix}
	\rr{Ch}_B(A_0,A_1,A_2)\;:=\;2\;\rr{Int}_B\big(\Gamma\;\pi(A_0)\;[D_B,\pi(A_1)]\;[D_B,\pi(A_2)]\big)\;,
\end{equation}
where $\rr{Int}_B$ is given by \eqref{eq:meas_prop_trip_02}. It turns out that $\rr{Ch}_B$ is well-defined on every triple
$A_0,A_1,A_2\in \bb{S}_B$. More precisely, one has that:
\begin{lemma}[Second Connes' formula\footnote{The name \emph{second Connes' formula} is borrowed from~\cite[Theorem 10]{bellissard-elst-schulz-baldes-94}. It is worth to point out that the  \emph{first Connes' formula} for the magnetic spectral triple has been proved in~\cite{denittis-sandoval-00}.}
		- version 1]\label{lem:cyc2cocy_Dix}
	It holds true that
	\begin{equation}
		\label{eq:Second-connes-formula-I}
		\rr{Ch}_B(A_0,A_1,A_2)\;=\;\frac{\ii}{\ell_B^2}\; \Psi_B(A_0,A_1,A_2)\;,\qquad\forall\; A_0,A_1,A_2\in \bb{S}_B
	\end{equation}
	with $\Psi_B$ given by~\eqref{eq:2-cocy-Psi}. As a consequence,
	$\rr{Ch}_B$ is a cyclic $2$-cocycle of $\bb{S}_B$.
\end{lemma}

The proof of Lemma \ref{lem:cyc2cocy_Dix} relies on a direct computation and the details are postponed to Section \ref{sec_prop_intro_1}. As a direct consequence of Lemma \ref{lem:cyc2cocy_Dix}, one gets that $\rr{Ch}_B$
provides a different representative for the class $[\Psi_B]\in HC^{\rm even}(\bb{S}_B)$, up to the \emph{right} constant.
In view of this observation, one can compute the {Chern number map} \eqref{eq:BM_02}  by using directly the cocycle $\rr{Ch}_B$.
\begin{theorem}[Chern number map]\label{teo:02}
	Let $c_B:K_0(\bb{S}_B)\to\Z$ be the  isomorphism defined in \eqref{eq:BM_02}. Then, it holds true that
	\begin{equation}
		\label{eq:teo-02}
		c_B([P])\;=\; \rr{Ch}_B(P,P,P)\;,\qquad [P]\in K_0(\bb{S}_B)\;
	\end{equation}
	where   $P\in \bb{S}_B$ is any representative of the class $[P]$.
\end{theorem}

It is worth to point out that the result contained in Theorem \ref{teo:02}  relates the topology of $\bb{S}_B$ with the
geometry of the spectral triple $(\bb{S}_B,\s{H}_4, D_B)$.

\begin{remark}[Involutions and topological triviality]\label{rk:triv_inv}
	By using the full set of $\gamma$ matrices one can construct the operator
	\begin{equation}\label{eq:inv_chi}
		\chi\;:=\;{\bf 1}\otimes \gamma_1\gamma_2\gamma_3\gamma_4\;=\;
		\left(\begin{array}{c c c c}
				-{\bf 1} & 0        & 0        & 0        \\
				0        & +{\bf 1} & 0        & 0        \\
				0        & 0        & -{\bf 1} & 0        \\
				0        & 0        & 0        & +{\bf 1}
			\end{array}
		\right)\;.
	\end{equation}
	Like $\Gamma$, this is also an involution, \ie\@ $\chi=\chi^*=\chi^{-1}$. Moreover, $\Gamma$
	anti-commutes with the Dirac operator, \ie\@ $\chi D_B=-D_B \chi$. This makes $(\bb{S}_B,\s{H}_4, D_B,\chi)$ an \emph{even} spectral triple~\cite[Section 3.1]{denittis-sandoval-00}. The latter property is not shared by the involution $\Gamma$. In fact, an easy calculation shows that
	$\Gamma D_B\neq-D_B \Gamma$ (see Appendix \ref{sec:q-sim-dir}). In view of this consideration, it would seem natural to consider the involution $\chi$ instead $\Gamma$ in the construction of the 2-cocycle \eqref{eq:2-cocy-Psi-Dix}. However, if one defines the quantity
	\begin{equation}\label{eq:2-cocy-Psi-Dix-triv}
		\widehat{\rr{Ch}_{B}}(A_0,A_1,A_2)\;:=\;2\;\rr{Int}_B\big(\chi\;\pi(A_0)\;[D_B,\pi(A_1)]\;[D_B,\pi(A_2)]\big)\;,
	\end{equation}
	then the argument  described in Remark \ref{rk:triv_inv_2}  provides
	\begin{equation}\label{eq:2-cocy-Psi-Dix-triv02}
		\widehat{\rr{Ch}_B}(P,P,P)\;=\;0\;,\qquad\forall\; P\in \bb{S}_B\;.
	\end{equation}
	The triviality expressed by equation \eqref{eq:2-cocy-Psi-Dix-triv02} has a deeper motivation. In fact the even spectral triple $(\bb{S}_B,\s{H}_4, D_B,\chi)$ turns out to be a representative of the trivial class in the KK-homology of $\bb{S}_B$~\cite{bourne-priv}. \hfill $\blacktriangleleft$
\end{remark}

Theorem \ref{teo:02} suggests the possibility of expressing the Chern number map  \eqref{eq:BM_02} inside the theory of the \emph{quantized calculus}~\cite[Chapter IV]{connes-94} associated with the magnetic spectral triple $(\bb{S}_B,\s{H}_4, D_B)$. However, as suggested by Remark \ref{rk:triv_inv} it is not the right choice to consider the latter as an even spectral triple with respect to the involution $\chi$. Moreover, definition \eqref{eq:2-cocy-Psi-Dix} shows that an important role is played by the involution $\Gamma$. All these reasons lead to develop the quantized calculus for the \emph{quasi-even} (\cf Definition \ref{def:quas-even}) magnetic spectral triple  $(\bb{S}_B,\s{H}_4, D_B, \Gamma)$. This will be done in full detail in Section \ref{sect:quan-cal}. In order to anticipate the main results let us introduce the
	{\em Dirac phase}
\begin{equation}\label{eq:dir_phas}
	F_{B,\varepsilon} \;:=\; \frac{D_{B}}{\lvert D_{B,\varepsilon} \rvert }\;,\qquad \varepsilon>0\;,
\end{equation}
the \emph{quasi-differential} (\cf Section \ref{sect:q-dif-str})
\begin{equation}\label{eq:dir_dif}
	\dd_B T\;:=\;[F_{B,\varepsilon},T] \;=\; F_{B,\varepsilon}T\;-\;TF_{B,\varepsilon}\;,
\end{equation}
which, in principle, is well-defined for every bounded operator $T\in\bb{B}(\s{H}_4)$,
and the \emph{compatible graded trace} (\cf Definition \ref{def:qg_tr})
\begin{equation}\label{eq:dir_trac}
	\rr{tr}_\Gamma(T)\;:=\;\Tr_{\rm Dix}\left(\Gamma T\right)\;,
\end{equation}
which is well-defined whenever $T\in \rr{S}^{1^+}_{\rm m}$
Then, it follows that the
	{compatible graded trace} $ \rr{tr}_\Gamma$  and the quasi-differential $\dd_B$  provide the constitutive elements of a \emph{quasi-cycle} of dimension 2   for the smooth magnetic algebra $\bb{S}_B$.
This concept will be clarified in full detail in Section \ref{sect:q-C}. Associated with the two-dimensional quasi-cycle there is a canonical
\emph{character} defined by
\begin{equation}\label{eq:ass_charact}
	{\tau}_{B,2}(A_0,A_1,A_2)\;:=\;\frac{1}{2}\;\mathfrak{tr}_\Gamma\big(\pi (A_0)\dd_B\pi(A_1)\dd_B\pi(A_2)\big)\;,
\end{equation}
which turns out to be well-defined for every $A_0,A_1,A_2\in \bb{S}_B$. Interestingly,
the character $ {\tau}_{B,2}$ identifies with $\Psi_B$ as showed
in the following  result, whose proof is postponed to Section \ref{sec_prop_intro_1}.
\begin{lemma}[Second Connes' formula - version 2]\label{lem:cyc2cocy_Dix-fredh}
	It holds true that
	\begin{equation}
		\label{eq:lem-cyc2cocy_Dix-fredh}
		\tau_{B,2}(A_0,A_1,A_2)\;=\;\frac{\ii}{\ell_B^2}\;\Psi_B(A_0,A_1,A_2)\;,\qquad\forall\; A_0,A_1,A_2\in \bb{S}_B
	\end{equation}
	with $\Psi_B$ given by \eqref{eq:2-cocy-Psi}. As a consequence,
	$ {\tau}_{B,2}$ is a cyclic $2$-cocycle of $\bb{S}_B$.
\end{lemma}
Since $ {\tau}_{B,2}$ is a  cyclic $2$-cocycle of $\bb{S}_B$, it defines a class in the cyclic cohomology which is usually denoted as
$Ch_2(\s{H}_4,F_{B,\varepsilon}):=[ {\tau}_{B,2}]\;\in\;HC^{\rm even}(\bb{S}_B)$. According to the common use, we will refer to  $Ch_2(\s{H}_4,F_{B,\varepsilon})$ as the  \emph{Chern character} of the \emph{quasi-even} Fredholm module $(\s{H}_4,F_{B,\varepsilon})$ endowed with the involution $\Gamma$. As a consequence of Lemma \ref{lem:cyc2cocy_Dix-fredh} and Lemma \ref{lem:cyc2cocy_Dix}
one obtains the following restatement of   Theorem \ref{teo:02}.
\begin{theorem}[Chern character]\label{teo:03}
	The isomorphism $c_B:K_0(\bb{S}_B)\to\Z$  defined by \eqref{eq:BM_02} provides the pairing between $K_0(\bb{S}_B)$ and the {Chern character} $Ch_2(\s{H}_4,F_{B,\varepsilon})$ of the Fredholm module $(\s{H}_4,F_{B,\varepsilon})$, \ie
	\begin{equation}
		\label{eq:teo-03}
		c_B([P])\;=\;\langle Ch_2(\s{H}_4,F_{B,\varepsilon}),[P]\rangle\;,\qquad [P]\in K_0(\bb{S}_B)\;.
	\end{equation}
\end{theorem}
\begin{corollary}[The Connes-Kubo-Chern formula]
	Let $H$ be a  self-adjoint operator affiliated with the magnetic von Neumann
	algebra $\bb{M}_B$. Assume that the spectrum of $H$ is bounded from below, and
	that for every (Fermi) energy $E\in\rho(H)$ in the resolvent set of $H$ the spectral projection $P_E:=\chi_{(-\infty,E)}(H)$ lies in
	$\bb{S}_B$.
	Then the \emph{Hall conductance} associated to  the energy spectrum of $H$ below the (Fermi) energy $E$
	is given
	by
	\[
		\sigma_{H}(E)\;=\;\frac{e^2}{2\pi \hbar}\; \langle Ch_2(\s{H}_4,F_{B,\varepsilon}),[P_E]\rangle\;.
	\]
\end{corollary}

\begin{remark}[Compatible graded trace and noncommutative integral]\label{rk:vol_form_diff}
	It is worth spending some words about a comparison between the
		{noncommutative integral} $\rr{Int}_B$ defined by
	\eqref{eq:meas_prop_trip_02} and the
		{compatible graded trace} $\rr{tr}_\Gamma$ defined by
	\eqref{eq:dir_trac}. Both are built by means of  the Dixmier trace $\Tr_{\rm Dix}$ but in $\rr{Int}_B$ the  Dixmier trace is weighted by the term $|D_{B,\varepsilon}|^{-2}$ which plays the role of a (noncommutative)  \emph{infinitesimal element of volume}.
	From  Lemma \ref{lem:cyc2cocy_Dix} and Lemma \ref{lem:cyc2cocy_Dix-fredh} one infers the equality ${\tau}_{B,2}=\rr{Ch}_B$.
	However,  ${\tau}_{B,2}$ is defined in terms of  $\rr{tr}_\Gamma$ while $\rr{Ch}_B$ is constructed with $\rr{Int}_B$. Nevertheless the equality between the two $2$-cocycles is made possible since the {quasi-differential} $\dd_B$ which enters in the construction of ${\tau}_{B,2}$ provides a weight proportional to $|D_{B,\varepsilon}|^{-1}$, which is exactly the square root of the  infinitesimal element of volume.
	\hfill $\blacktriangleleft$
\end{remark}

The novelty of the results contained in Theorem \ref{teo:01},
Theorem \ref{teo:02}, and Theorem \ref{teo:03} consists on the
use of the magnetic spectral triple $(\bb{S}_B,\s{H}_4, D_B)$, or in the associated Fredholm module $(\s{H}_4,F_{B,\varepsilon})$, for the study of the topology of the magnetic algebra $\bb{S}_B$.
The relevance of this approach relies on the fact that the Dirac operator $D_B$, as defined by \eqref{eq:intro_D}, has compact resolvent. Equivalently, the Dirac phase $F_{B,\varepsilon}$ is a compact operator. This compactness is the real new insight of our approach to the study of the magnetic algebra, which indeed contrasts with other approaches already present in the literature.
For a more precise analysis on this aspect we refer to the
long discussion contained in~\cite[Section 1]{denittis-sandoval-00} and references therein. In order to advocate for the usefulness of the compactness in our approach, let us rewrite the integrated density of states in \eqref{eq:int_idos} and the Hall conductance in \eqref{eq:int_hall_cond} in combination with the results of
Theorem
\ref{teo:01} and Theorem \ref{teo:03}. By making explicit the role of the Dixmier trace and of the resolvent of $D_B$ in the definition of the noncommutative integral $\rr{Int}_B$, one obtains
\begin{equation}\label{eq:int_idos_2}
	N_H(E)\;=\; \frac{1}{8\pi\ell^2_B}\;\Tr_{\rm Dix}\left(|D_{B,\varepsilon}|^{-2}\;\pi(P_E)\right)
\end{equation}
for the {integrated density of states},  and
\begin{equation}\label{eq:int_hall_cond_2}
	\begin{aligned}
		\sigma_H(E)\; & =\;\frac{ e^2}{4\pi \hslash }\;\Tr_{\rm Dix}\big(\Gamma\;\pi (P_E)\dd_B\pi(P_E)\dd_B\pi(P_E)\big)                 \\
		              & =\;\frac{ e^2}{4\pi \hslash }\;\Tr_{{\rm Dix}}\left(|D_{B,\varepsilon}|^{-2}\Gamma\pi(P_E)[D_B,\pi(P_E)]^2\right)
	\end{aligned}
\end{equation}
for the  Hall conductance. Since the operator $|D_{B,\varepsilon}|^{-2}$ is diagonalized by the  Laguerre
basis
$\{\psi_{n,m}\}$, one can hope to use this natural \virg{discretization} to deduce from  \eqref{eq:int_idos_2} and \eqref{eq:int_hall_cond_2} approximate formulas for $N_H(E)$ and
$\sigma_H(E)$. In the case of tight-binding magnetic operators on $\ell^2(\Z^2)$ similar approximated formulas already exists, based on the discreteness of the lattice $\Z^2$. In fact the  density of states for tight-binding magnetic operators
can be estimated with the \emph{windowed DOS}~\cite{loring-lu-watson-21} while
the Chern numbers can be computed with  the \emph{spectral localizer formula}~\cite{schulz-baldes-loring-17,schulz-baldes-loring-19,schulz-baldes-loring-20}.
Our guess is that the latter results can be adapted to the magnetic operators on $L^2(\R^2)$ on the basis of the formulas  \eqref{eq:int_idos_2} and \eqref{eq:int_hall_cond_2}. At the moment, this idea is under investigation.

\medskip

\noindent
{\bf Structure of the paper.}
In Section~\ref{sec:QE-FM} we introduce a generalization of an even Fredholm
module, these so-called \emph{quasi-even Fredholm modules}, which will be used to study the
differential theory of the magnetic algebra. The interest of this generalization
lies in the fact that quotient by a convenient ideal of compact operators gives rise to a genuine even Fredholm module. In Section~\ref{sect:q-dif-str} we study the differential theory of
quasi-even Fredholm modules, with the goal of defining the notion of a $k$-cycle
over a quasi-even Fredholm modules in Section~\ref{sect:q-C}. Here we also introduce an appropriate notion of a graded trace compatible with
quasi-even Fredholm modules.
In Section~\ref{sec:chern-character} we identify the Chern character of the
$2$-cycle associated to the magnetic algebra to~\eqref{eq:2-cocy-Psi} via the second Connes' Formula.
Section \ref{sec_prop_intro_1} contains the
proofs of the key Lemmas \ref{lem:cyc2cocy_Dix} and \ref{lem:cyc2cocy_Dix-fredh}.
In Appendix \ref{app:tech_res} are collected some technical results used in various parts of the paper. Appendix \ref{app:cyc_cohom} is devoted to a brief overview of the cyclic cohomology of the  magnetic algebra.

\medskip

\noindent{\bf Acknowledgments.}
GD's research is supported by the grant {Fondecyt Regular - 1190204}. MS’s research is supported by the grant
	{CONICYT-PFCHA Doctorado Nacional 2018 - 21181868}.
The authors
would like to cordially thank Chris Bourne and Hermann
Schulz-Baldes for several inspiring discussions.

\section{Quantized calculus of the magnetic spectral triple}\label{sect:quan-cal}
In this section we will build the \emph{quantized calculus}
(a la Connes) for the algebra $\bb{S}_B$ based on the geometry of the {magnetic spectral triple} $(\bb{S}_B,\s{H}_4, D_B)$.
Although the treatment presented below follows quite closely the theory presented in~\cite[Chapter IV]{connes-94} we will need to change and generalize some definitions to adapt the general scheme to our case of interest.

\subsection{Quasi-even Fredholm module}\label{sec:QE-FM}
Let $\bb{K}(\s{H})$ be the $C^*$-algebra of compact operators on a Hilbert space $\s{H}$. Let $\bb{A}$ be a pre-$C^*$-algebra and $\pi:\bb{A}\to\bb{B}(\s{H})$ a $\ast$-representation. Following~\cite[Chapter IV]{connes-94}, let us recall that a (compact) Fredholm module over $\bb{A}$, denoted $(\s{H},F)$, is   determined by a bounded operator $F$ such that: $(F-F^*)\in \bb{K}(\s{H})$ (quasi-self-adjoint); $(F^2-{\bf 1})\in \bb{K}(\s{H})$ (quasi-involution) and
$$
	[F,\pi(A)]\;:=\;F\;\pi(A)\;-\;\pi(A)\; F\;\in\; \bb{K}(\s{H})\;,\qquad \forall\; A\in \bb{A}\;.
$$
A graded structure on $\s{H}$ is given by a self-adjoint non-trivial\footnote{That is $\Gamma\neq \pm{\bf 1}$. Equivalently, the spectrum of $\Gamma$ is $\{\pm 1\}$.} involution
$\Gamma=\Gamma^*=\Gamma^{-1}$. A bounded operator $T\in \bb{B}(\s{H})$ has \emph{degree} $0$ with respect to $\Gamma$ if $\Gamma T=T \Gamma$, and has \emph{degree} $1$ if $\Gamma T=-T \Gamma$. We will denote with $\bb{B}(\s{H})_i$ the subset of bounded operator of degree $i=0,1$. In order to combine a graded structure with a Fredholm module the basic request is that the representation $\pi$ has to be of {degree} $0$, \ie\@ $\pi(\bb{A})\subseteq\bb{B}(\s{H})_0$. Said differently, one requires that
$$
	\Gamma\;\pi(A)\;-\;\pi(A)\; \Gamma\;=\;0\;,\qquad \forall\; A\in \bb{A}\;.
$$
A Fredholm module $(\s{H},F)$ with graded structure $\Gamma$ is called \emph{even} if $F\in  \bb{B}(\s{H})_1$, \ie when
\begin{equation}\label{eq:even-gamma}
	\{\Gamma,F\}\;:=\;\Gamma\;F\;+\;F\;\Gamma\;=\;0
\end{equation}
For our aim, equation \eqref{eq:even-gamma}
is not satisfied (\cf Remark \ref{rk:triv_inv}) and for this reason
we need to adapt the notion of {even} Fredholm module.
\begin{definition}[Quasi-even   Fredholm module of dimension $k$]\label{def:quas-even}
	Let $\rr{Z}\subseteq \bb{K}(\s{H})$ be a two-sided self-adjoint ideal of $\bb{B}(\s{H})$ and $\Gamma$ a non-trivial self-adjoint involution. A
	Fredholm module $(\s{H},F)$ over $\bb{A}$ is called
	\emph{quasi-even of dimension $k$} with respect to the pair $(\Gamma,\rr{Z})$ if:
	\begin{itemize}
		\item[(a)] $[F^2,\pi(A_0)]\in\rr{Z}$ for every $A_0\in \bb{A}$;
			\vspace{1mm}
		\item[(b)] $\Gamma [F,\pi(A_0)]\Gamma=-[F,\pi(A_0)]+R(A_0)\Gamma$ such that
			$$
				[F,\pi(A_1)]\;\cdots\;[F,\pi(A_n)]\;R(A_0)\;[F,\pi(A_{n+1})]\cdots[F,\pi(A_{k-1})]\;\in\;\rr{Z}\;,
			$$
			for every
			$A_0,A_1,\ldots,A_n,A_{n+1},\ldots,A_{k-1}\in \bb{A}$.
			\vspace{1mm}
		\item[(c)] $[F,\pi(A_0)][F,\pi(A_1)] \cdots[F,\pi(A_k)] \in\rr{Z}$
			for every $A_0,A_1,\ldots,A_k\in \bb{A}$.
	\end{itemize}
\end{definition}

Condition (a) of Definition \ref{def:quas-even} generalizes the requirement
$F^2={\bf 1}$, which is usually assumed in the theory of Fredholm modules
(see~\cite[Remark 3.13]{denittis-sandoval-00} and references therein).
It is immediate to observe that every even Fredholm module
meets  condition (b) of Definition~\ref{def:quas-even} with $R=0$. In this sense Definition~\ref{def:quas-even} provides a generalization of the notion of  {even} Fredholm module.
From condition (b) one deduces that
$$
	\begin{aligned}
		 & \Gamma\; [F,\pi(A_1)]\; \cdots\;[F,\pi(A_k)]                                                                    \\
		 & =\;(-1)^{n}\;[F,\pi(A_1)]\; \cdots \;[F,\pi(A_{n})]\;\Gamma\;[F,\pi(A_{n+1})]\; \cdots\;[F,\pi(A_k)]\;+\;\rr{Z}
	\end{aligned}
$$
for all $A_1,\ldots,A_k\in \bb{A}$.
When $n=k$ this implies
\begin{equation}\label{eq-parity-k}
	[F,\pi(A_1)]\;\cdots\;[F,\pi(A_k)]\;\in\;\bb{B}(\s{H})_{k\;{\rm mod}\; 2}\;+\;\rr{Z}
\end{equation}
Condition (c)  stipulates that
\begin{equation}\label{eq-parity-k2}
	[F,\pi(A_1)]\;\cdots\;[F,\pi(A_{k'})]\; \in\;\rr{Z}\;
\end{equation}
as soon as $k'>k$ since $\rr{Z}$ is an ideal.

\medskip

Now, let us focus on the \emph{magnetic Fredholm module} $(\s{H}_4,F_{B,\varepsilon})$ over the pre-$C^*$-algebra $\bb{S}_B$, where
the Hilbert space $\s{H}_4$ is defined by \eqref{eq:H4} and the Dirac phase is defined by \eqref{eq:dir_phas}. Let $\Gamma$ be the self-adjoint involution defined by \eqref{eq:inv_Gamma}.
Finally, let us recall the notations $\rr{S}^{p}$ and   $\rr{S}^{p^\pm}$ for the $p$-th Schatten ideal and for the $p$-th Dixmier/Ma\u{c}aev  ideal, respectively (\cf Appendix~\ref{app_w-lp-dual}). The main properties of the {magnetic Fredholm module} $(\s{H}_4,F_{B,\varepsilon})$ are contained in the following result.
\begin{proposition}\label{prop:prop-MFM}
	The following facts hold true:
	\begin{itemize}
		\item[(1)] $[F_{B,\varepsilon},\pi(A)]\in \rr{S}^{2^+}$ for every $A\in \bb{S}_B$;
			\vspace{1mm}
		\item[(2)] $[F_{B,\varepsilon}^2,\pi(A)]\in \rr{S}^{1}$ for every $A\in \bb{S}_B$;
			\vspace{1mm}
		\item[(3)] Let $R(A_0):=
				\Gamma[F_{B,\varepsilon},\pi(A_0)]\Gamma+[F_{B,\varepsilon},\pi(A_0)]$, then
			$$
				R(A_0)\;[F_{B,\varepsilon},\pi(A_1)]\;\in\;  \rr{S}^{1}\;, \quad[F_{B,\varepsilon},\pi(A_1)]\; R(A_0)\;\in\;  \rr{S}^{1}
			$$
			for every $A_0,A_1\in \bb{S}_B$;
			\vspace{1mm}
		\item[(4)]  $[F_{B,\varepsilon},\pi(A_0)][F_{B,\varepsilon},\pi(A_1)][F_{B,\varepsilon},\pi(A_2)]\in \rr{S}^{1}$ for every $A_0,A_1,A_2\in \bb{S}_B$;
	\end{itemize}
\end{proposition}
\begin{proof}
	Item (1) is proved in~\cite[Lemma 3.10]{denittis-sandoval-00}. Item (2)
    follows from the direct computation
	\begin{equation}
		F_{B,\varepsilon}^2\;-\;{\bf 1}\;=\;-\varepsilon\; |D_{B,\varepsilon}|^{-2}
	\end{equation}
	which provides
	$$
		[F_{B,\varepsilon}^2,\pi(A)]\;=\;-\varepsilon\;\left[|D_{B,\varepsilon}|^{-2},\pi(A)\right]\;=\;-\varepsilon\;\sum_{j=1}^4\left[Q_{B,\varepsilon_j}^{-1},A\right]\;\otimes\;\tau_{j,j}
	$$
	where
	$Q_{B,\varepsilon_j}:=Q_B+\varepsilon_j{\bf 1}$
	with $Q_B$  the {harmonic oscillator} \eqref{eq:harm_osc},
	$$
		\varepsilon_1\;=\;\varepsilon_2\;=\;\varepsilon\;,\quad \varepsilon_3\;=\;\varepsilon+1\;,\quad \varepsilon_4\;=\;\varepsilon-1\;,
	$$
	and $\tau_{i,j}\in{\rm Mat}_4(\C)$ the matrix which has a single 1 in the entry at
	the position $(i,j)$ and zeroes in all other positions. Therefore, to prove the
	result it is enough to shows that  $[Q_{B,\varepsilon_j}^{-1},A]\in \rr{S}^{1}$ for every
	$j=1,\ldots,4$ and this is done in Lemma \ref{lemma:help_ideal} and Remark
	\ref{rk:trac-class-strong}. Item (3) follows from Lemma
	\ref{lemma:help_ideal_basic} which shows that $R(A_0)\in \rr{S}^{2^-}$. Since
	$\rr{S}^{2^+}$ is the dual of $\rr{S}^{2^-}$ one gets from item (1) and
	Corollary \ref{cor_dual_opo} the desired result. For item (4) one needs to use
	the H\"older type inequality for weak Schatten ideals $\rr{S}^{p}_{\rm w}$ (see Appendix \ref{app_w-lp-dual}).
	since $[F_{B,\varepsilon},\pi(A_j)]\in \rr{S}^{2^+}=\rr{S}^{2}_{\rm w}$ one obtains that  triple products of these terms lie inside
	$\rr{S}^{\frac{2}{3}}_{\rm w}$. The inclusion $\rr{S}^{\frac{2}{3}}_{\rm w}\subset \rr{S}^1$ concludes the proof.
\end{proof}

Item (1) of Proposition \ref{prop:prop-MFM} says that
the  magnetic Fredholm module $(\s{H}_4,F_{B,\varepsilon})$ is
(densely)  \emph{$2^+$-summable} (\cf~\cite[Theorem 3.12]{denittis-sandoval-00}).
Summarizing all the previous results we can state that:
\begin{theorem}
	The magnetic Fredholm module $(\s{H}_4,F_{B,\varepsilon})$ over the pre-$C^*$-algebra $\bb{S}_B$ is (densely)  {$2^+$-summable} and
	quasi-even of rank 2 with respect to the pair $(\Gamma, \rr{S}^{1})$.
\end{theorem}

\subsection{Quasi-differential structure}
\label{sect:q-dif-str}
Let $(\s{H},F)$ be a Fredholm module over $\bb{A}$
with a quasi-even structure of dimension $k$ with respect to $(\Gamma, \rr{Z})$. Let
$$
	{\Omega}^0\;:=\;\pi(\bb{A})^+\;=\;\big\{\pi(A)+c{\bf 1}\in\bb{B}(\s{H})\;\big|\; A\in \bb{A},\;\; c\in\C\big\}\;.
$$
Observe that ${\Omega}^0=\pi(\bb{A})$ whenever $\bb{A}$ is unital and $\pi({\bf 1})={\bf 1}$. Moreover, ${\Omega}^0_B$ is made by element of degree $0$ with respect to $\Gamma$, \ie
${\Omega}^0\subseteq\bb{B}(\s{H})_0$.
The \emph{quasi-differential} on  ${\Omega}^0$
is defined by
\begin{equation}\label{eq:q-diff-01}
	\dd \big(\pi(A)+c{\bf 1}\big)\;:=\;[F,\pi(A)]
\end{equation}
for every $A\in\bb{A}$.
For $n\in\N$ one lets ${\Omega}^n$ be the linear span of elements of the type
$$
	\omega\;:=\;(\pi(A_0)+c{\bf 1})\;\dd \pi(A_1)\;\cdots\:\dd \pi(A_n)\;,\qquad A_0,A_1,\ldots, A_n\in \bb{A}\;.
$$
In short, one can write
${\Omega}^n:=\pi(\bb{A})^+\otimes\dd \pi(\bb{A})\otimes\ldots\otimes \dd \pi(\bb{A})$ where the product is repeated $n$-times. From \eqref{eq-parity-k} it follows that
$$
	{\Omega}^{n}\;\subseteq\; \bb{B}(\s{H})_{n\;{\rm mod}\; 2}\;+\;\rr{Z},\qquad n \in \mathbb{N}_{0}\;
$$
and \eqref{eq-parity-k2} implies that
\begin{equation}\label{eq:id_incl}
	{\Omega}^{n}\;\subseteq\; \rr{Z}\;,\qquad \forall\; n>k\;.
\end{equation}
The full \emph{graded quasi-differential algebra} associated with $\bb{A}$ is defined as
$$
	{\Omega}^\bullet\;:=\;\bigoplus_{n\in\N_0}{\Omega}^n\;.
$$
In ${\Omega}^\bullet_B$, the product of operators provides a   well-defined graded product. More precisely, for every pair of elements $\omega\in {\Omega}^n$ and $\omega'\in {\Omega}^{n'}$ one has that  $\omega\omega'\in {\Omega}^{n+n'}$. The proof of this fact is  straightforward (see~\cite[Section IV.1]{connes-94}) and is based on the identity
$$
	[F,\pi(A_1)]\pi(A_2)\;=\;[F,\pi(A_1A_2)]\;-\;\pi(A_1)[F,\pi(A_2)]
$$
valid for every $A_1,A_2\in\bb{A}$.

\medskip

The quasi-differential \eqref{eq:q-diff-01} can be extended to a map $\dd:{\Omega}^\bullet\to{\Omega}^\bullet$ as follows
\begin{equation}\label{eq:q-diff-01}
	\dd \omega\;:=\;F\;\omega\;-\;(-1)^n\;\omega\;F\;,\qquad \forall\;\omega\in{\Omega}^n\;.
\end{equation}
the main  properties of the quasi-differential are listed below.
\begin{proposition}\label{prop:QD}
	Let $(\s{H},F)$ be a Fredholm module over $\bb{A}$
	with a quasi-even structure of dimension $k$ with respect to $(\Gamma, \rr{Z})$.
	The following facts hold true:
	\begin{itemize}
		\item[(1)] $\dd\omega\in \Omega^{n+1}+\rr{Z}$ for every $\omega\in  \Omega^{n}$;
			\vspace{1mm}
		\item[(2)] $\dd^2\omega:=\dd(\dd\omega)\in \rr{Z}$ for every $\omega\in  \Omega^{n}$;
			\vspace{1mm}
		\item[(3)] $\dd(\omega_1\omega_2)=(\dd \omega_1)\omega_2+(-1)^{n_1}\omega_1(\dd \omega_2)$ for every $\omega_j\in\Omega^{n_j}$, with $j=1,2$.
	\end{itemize}
\end{proposition}

\begin{proof}
	Item (1) follows from the identity
	$$
		[F,\pi(A)]\;F\;=\;-F\;[F,\pi(A)]\;+\;Z_A
	$$
	where $Z_A:=[F^2,\pi(A)]\in\rr{Z}$ by assumption. Therefore, for every $A_1,...,A_n\in\bb{A}$ one gets
	$$
		\big([F,\pi(A_1)]\;\ldots\; [F,\pi(A_n)]\big)\;F\;=\;(-1)^n\;F\;\big([F,\pi(A_1)]\;\ldots\; [F,\pi(A_n)]\big)\;+\;Z
	$$
	for a certain $Z\in\rr{Z}$ which depends on $A_1,...,A_n$.
	Let
	\begin{equation}\label{eq:el_eta}
		\eta\;:=\;(\pi(A_0)+c{\bf 1})\;\dd \pi(A_1)\;\ldots\;\dd \pi(A_n)
	\end{equation}
	be one of the  elemental generators of $\Omega^n$.
	Then, it turns out that
	$$
		\eta\;F\;-\;(-1)^n\;F\;\eta\;=\;[F, \pi(A_0)+c{\bf 1}]\; \dd \pi(A_1)\;\ldots\:\dd \pi(A_n)\;+\;Z'
	$$
	with $Z'\in \rr{Z}$. This, proves that $\dd \eta\in \Omega^{n+1}+\rr{Z}$. Since every $\omega\in \Omega^{n}$ is a linear combination of elements of the type of $\eta$, the result follows by linearity.
	Item (2) follows by a direct computation which shows that
	$$
		\dd^2\omega\;=\;[F^2,\omega]\;,\qquad \forall\;\omega\in\Omega^n\;.
	$$
	If $\omega\in\Omega^0$ then $\dd^2\omega\in \rr{Z}$ just by assumption. To complete the proof one can use induction on the
	order $n$. Let us assume that   item (2) is true up to order $n-1$ and consider an element
	$\eta\in\Omega^n$ defined as in
	\eqref{eq:el_eta}. One has that $\eta=(\pi(A_0)+c{\bf 1})\eta_0$
	where $\eta_0:=\dd \pi(A_1)\ldots\dd \pi(A_n)\in \Omega^{n-1}$
	Therefore,
	$$
		\dd^2\eta\;=\;[F^2,\pi(A_0)]\;\eta_0 \;+\;(\pi(A_0)+c{\bf 1})\;[F^2,\eta_0]\;\in\;\rr{Z}
	$$
	since both $[F^2,\pi(A_0)]$ and $[F^2,\eta_0]$ are in $\rr{Z}$ by assumption. By linearity one gets the result for a generic $\omega\in\Omega^n$.
	The proof of Item (3) amounts to a direct computation.
\end{proof}

From items (1) and (2) of Proposition \ref{prop:QD} one infers that the quasi-differential $\dd$ acts on the graded algebra ${\Omega}^\bullet$ as follows:
$$
	\begin{aligned}
		\dd\;   & :\;{\Omega}^{n}\;\longrightarrow\;{\Omega}^{n+1}\;+\;{\rr{Z}} \\
		\dd^2\; & :\;{\Omega}^{n}\;\longrightarrow\;{\rr{Z}}
	\end{aligned}\;,\qquad \forall\, n\in\N_0\;.
$$
Moreover, item (3) shows that $\dd$ is a \emph{graded derivation}, \ie\@ it satisfies a graded version of the Leibniz's rule.

\begin{remark}[Induced differential structure]\label{rk:quot-grad-alg} Consider the quotient space
	$
		\widetilde{\Omega}^{n}:={\Omega}^{n}/\rr{Z}
	$
	and the related graded algebra
	$$
		\widetilde{\Omega}^\bullet\;:=\;\bigoplus_{n\in\N_0}\widetilde{\Omega}^n\;.
	$$
	From Proposition \ref{prop:QD} one infers that the quasi-differential $\dd$ behaves well with respect to the quotient an defines a map
	$$
		\begin{aligned}
			\tilde{\dd}\; & :\;\widetilde{\Omega}^{n}\;\longrightarrow\;\widetilde{\Omega}^{n+1}\;,\qquad \forall\, n\in\N_0
		\end{aligned}\;
	$$
	such that $\tilde{\dd}^2=0$.  Said differently, the pair $(\widetilde{\Omega}^\bullet,\tilde{\dd})$ defines a genuine \emph{graded differential algebra} in the sense described in
	\cite[Section IV.1]{connes-94}. It is worth noting that according to  Definition \ref{def:quas-even} one has that $\widetilde{\Omega}^{n}=0$ for every $n>k$ in the case of a quasi-even
	structure of dimension $k$.
	\hfill $\blacktriangleleft$
\end{remark}

Let us now focus on the specific case of the magnetic Fredholm module $(\s{H}_4,F_{B,\varepsilon})$.
We will denote with $\dd_B$ the quasi-differential associated with
$F_{B,\varepsilon}$ according to the definition \eqref{eq:dir_dif}.
The related {graded quasi-differential algebra} associated with the magnetic algebra $\bb{S}_B$ will be denoted with
$$
	{\Omega}^\bullet_B\;:=\;\bigoplus_{n\in\N_0}{\Omega}^n_B\;.
$$
As a consequence of Proposition \ref{prop:prop-MFM} one has that
${\Omega}^1_B\subseteq \rr{S}^{2^+}$ and ${\Omega}^2_B\subseteq \rr{S}^{1^+}$.
Moreover ${\Omega}^2_B \subseteq  \bb{B}(\s{H}_4)_0+\rr{S}^{1}$, \ie
the elements of ${\Omega}^2_B$ are of degree 0 with respect to $\Gamma$ up to a remainder which is trace class.

\subsection{Quasi-cycles}\label{sect:q-C}
Let   $(\s{H},F)$ be a Fredholm module over $\bb{A}$
with a quasi-even structure of dimension $k$ with respect to $(\Gamma, \rr{Z})$. We need to consider a linear map
\begin{equation}\label{eq:qt}
	\rr{tr}_\Gamma\;:\;\Omega^k\;\longrightarrow\;\C
\end{equation}
which satisfies some relevant conditions.
\begin{definition}[A compatible graded trace]\label{def:qg_tr}
	Let $(\s{H},F)$ be a Fredholm module over $\bb{A}$
	with a quasi-even structure of dimension $k$ with respect to $(\Gamma, \rr{Z})$. A compatible graded trace for $(\s{H},F)$
	is a map like \eqref{eq:qt} such that :
	\begin{itemize}
		\item[(a)] $\rr{Z} \subseteq{\rm Ker}(\rr{tr}_\Gamma)$;
			\vspace{1mm}
		\item[(b)] $\rr{tr}_\Gamma(\dd(\omega))=0$ for every $\omega\in \Omega^{k-1}$;
			\vspace{1mm}
		\item[(c)] $\rr{tr}_\Gamma(\omega_1\omega_2)=(-1)^{n_1n_2}\rr{tr}_\Gamma(\omega_2\omega_1)$ for every $\omega_1\in \Omega^{n_1}$ and $\omega_2\in \Omega^{n_2}$  such that $n_1+n_2=k$.
	\end{itemize}
\end{definition}

Property (b) is the \emph{closedness} condition of the trace $\rr{tr}_\Gamma$ with respect to the quasi-differential $\dd$. Property (c) implements the compatibility of $\rr{tr}_\Gamma$  with respect to the graded structure of ${\Omega}^\bullet$. Finally,
from \eqref{eq:id_incl} one infers that
$\rr{tr}_\Gamma(\Omega^n) = \{0\}$ for every $n>k$.

\medskip

The following definition generalizes the concept of cycle given in
\cite[Chapter~3, Section~1.$\alpha$]{connes-94} or~\cite[Definition~8.3]{gracia-varilly-figueroa-01}.
\begin{definition}[Quasi-cycle of dimension $k$]\label{def:qc_dim-k}
	Let $(\s{H},F)$ be a Fredholm module over $\bb{A}$
	with a quasi-even structure of dimension $k$ with respect to $(\Gamma, \rr{Z})$. Let
	${\Omega}^\bullet$ be the associated
	graded quasi-differential algebra with quasi-differential $\dd$ induced by $F$ and $\rr{tr}_\Gamma$ a compatible graded trace. Then the triple $({\Omega}^\bullet,\dd,\rr{tr}_\Gamma)$ will be
	called a \emph{quasi-cycle of dimension $k$} over $\bb{A}$.
\end{definition}

\begin{remark}[Induced cycle]\label{rk:quot-cyc}
	In the same spirit of Remark \ref{rk:quot-grad-alg} one can observe that a compatible graded trace behaves well with respect to the quotient with respect to $\rr{Z}$ and defines a trace
	$$
		\begin{aligned}
			\widetilde{\rr{tr}}_\Gamma\; & :\;\widetilde{\Omega}^{k}\;\longrightarrow\;\C\;.
		\end{aligned}\;
	$$
	In particular, one can check that the triple $(\widetilde{\Omega}^\bullet,\tilde{\dd},\widetilde{\rr{tr}}_\Gamma)$ defines a
	genuine cycle of dimension $k$ in the sense of
	\cite[Chap.~3, Sect.~1.$\alpha$]{connes-94} or~\cite[Def.~8.3 \&
		Def.~8.17]{gracia-varilly-figueroa-01}.
	\hfill $\blacktriangleleft$
\end{remark}

Now, let us focus on the case of the magnetic Fredholm module $(\s{H}_4,F_{B,\varepsilon})$. In such a case, a natural
candidate for  a compatible graded trace is the Dixmier trace.  In order to take care of the grading, we will define the compatible graded trace as in \eqref{eq:dir_trac}, \ie
\begin{equation}\label{eq:dir_trac-02}
	\rr{tr}_\Gamma(\omega)\;:=\;\Tr_{\rm Dix}\left(\Gamma \omega\right)\;,\qquad \forall\; \omega\in\Omega^2_B\;.
\end{equation}
It is worth pointing out
that there is no need to specify the dependence of the Dixmier trace on the choice of a scale-invariant state. In fact, as commented in
\cite[Remark 3.11]{denittis-sandoval-00} one can show that
$\Omega^2_B\subseteq \rr{S}^{1^+}_{\rm m}$ where $\rr{S}^{1^+}_{\rm m}\subset \rr{S}^{1^+}$ denotes the closed space of \emph{measurable elements} whose Dixmier trace  does not depend on the choice of any scale-invariant generalized limit. Let us introduce the subspace $\rr{S}^{1^+}_{\rm m}$ of element with a vanishing Dixmier trace
$$
	{\rm Ker}(\Tr_{{\rm Dix}})\;:=\;\left.\left\{T\in\rr{S}^{1^+}_{\rm m}\;\right|\;\Tr_{{\rm Dix}}(T)=0\right\}\;.
$$
As discussed in Appendix \ref{app_w-lp-dual}, one has that $\rr{S}^{1}\subset \rr{S}^{1^+}_0\subset {\rm Ker}(\Tr_{{\rm Dix}})$.

\begin{proposition}\label{prop:clos_grad_trac}
	The map \eqref{eq:dir_trac-02} is a compatible graded trace for the
	magnetic Fredholm module $(\s{H}_4,F_{B,\varepsilon})$.
\end{proposition}

\begin{proof}
	Property (a) of Definition \ref{def:qg_tr} is satisfied  since
	the ideal
	$\rr{S}^{1}$ of trace-class operators is contained in ${\rm Ker}(\Tr_{{\rm Dix}})$.
	To prove property (b) of Definition \ref{def:qg_tr} let us start with an element $\omega=(\pi(A_0)+c{\bf 1})\dd_B\pi(A_1)\in\Omega^1$.
	Then $\dd_B\omega=\dd_B\pi(A_0)\dd_B\pi(A_1)+Z$ with $Z\in \rr{S}^{1}$ and in turn
	$$
		\rr{tr}_\Gamma\left(\dd_B\omega\right)\;=\; \rr{tr}_\Gamma\big([F_{B,\varepsilon},\pi(A_0)][F_{B,\varepsilon},\pi(A_1)]\big)\;=\;0
	$$
	in view of Corollary \ref{cor_closeddNes}. The general case follows by linearity since $\Omega^1$ is   the  linear span of elements of the form of $\omega$.
	Property (c) of Definition \ref{def:qg_tr} can be shown with a direct computation. There are two possible cases. Let us start with the case $\omega_0\in \Omega^0_B$ and $\omega_2\in \Omega^2_B$. Then
	\begin{equation}\label{eq:q_grad_magn}
		\begin{aligned}
			\rr{tr}_\Gamma(\omega_0\omega_2)\; & =\;\Tr_{\rm Dix}\left(\Gamma \omega_0\omega_2\right)\;=\;\Tr_{\rm Dix}\left(\Gamma^2\omega_2\Gamma \omega_0\right) \\
			                                   & =\;\Tr_{\rm Dix}\left(\Gamma \omega_2 \omega_0\right)\;=\;\rr{tr}_\Gamma(\omega_2\omega_0)
		\end{aligned}
	\end{equation}
	where  the second equality follows in view of the cyclicity of the Dixmier trace and the identity $\Gamma^2={\bf 1}$, while the second follows from $\Gamma\omega_2\Gamma=\omega_2+ Z$ for some $Z\in \rr{S}^{1}$ in view of Proposition \ref{prop:prop-MFM} (3). The second case consists in $\omega,\omega'\in \Omega^1_B$.
	By linearity it is enough to prove the claim for elements of the type $\omega=\eta_0\dd_B\eta_1$,
	and $\omega'=\eta_0'\dd_B\eta_1'$ with $\eta_j,\eta_j'\in\pi(\bb{B}_B)^+$ and $j=0,1$.
	We can use the
	same strategy of the computation \eqref{eq:q_grad_magn}.
	The first step consists in justifying the equality
	$$
		\Tr_{\rm Dix}\left(\Gamma \omega\omega'\right)\;=\;\Tr_{\rm Dix}\left(\omega' \Gamma \omega\right)\;.
	$$
	This is true since $\Gamma \omega\omega'$ and $\omega' \Gamma \omega$ are both in the Dixmier ideal and have the same system of non-zero eigenvalues~\cite[Section 3.10, Theorem 5]{birman-solomjak-87}. Therefore, the results follows in view of the \emph{Lidskii's formula} for the Dixmier trace~\cite[Theorem 7.3.1]{lord-sukochev-zanin-12}.
	The second step consists in proving that
	$\Gamma \omega' \Gamma=-\omega'+R$ with a remainder $R$ such that $R\omega\in \rr{S}^{1}$. This follows again from Proposition \ref{prop:prop-MFM} (3). Summing up one gets
	$$
		\rr{tr}_\Gamma(\omega\omega')\;=\;\Tr_{\rm Dix}\left(\Gamma \omega\omega'\right)\;=\;-\Tr_{\rm Dix}\left(\Gamma\omega'\omega\right)\;=\;-\rr{tr}_\Gamma(\omega'\omega)
	$$
	and the proof is completed.
\end{proof}

As a consequence of Proposition \ref{prop:clos_grad_trac} one has that the triple $({\Omega}^\bullet_B, \dd_B, \rr{tr}_\Gamma)$ provides a quasi-cycle of dimension $2$ over the algebra $\bb{S}_B$ in the sense of Definition \ref{def:qc_dim-k}. We will refer to
$({\Omega}^\bullet_B, \dd_B, \rr{tr}_\Gamma)$ as the \emph{magnetic quasi-cycle}.

\subsection{Chern character}
\label{sec:chern-character}
The following definition is adapted from
\cite[Chapet~3, Section~1.$\alpha$]{connes-94}
or~\cite[Definition~8.17]{gracia-varilly-figueroa-01}.

\begin{definition}[Character of a quasi-cycle]
	Let $(\s{H},F)$ be a Fredholm module over $\bb{A}$
	with a quasi-even structure of dimension $k$ with respect to $(\Gamma, \rr{Z})$. Let  $({\Omega}^\bullet,\dd,\rr{tr}_\Gamma)$  be
	the associated {quasi-cycle (of dimension $k$)} over $\bb{A}$
	according to Definition \ref{def:qc_dim-k}.
	The \emph{character} of the quasi-cycle is the $(k+1)$-linear functional $\tau_k:\bb{A}^{k+1}\to\C$ defined by
	$$
		\tau_k(A_0,A_1,\ldots,A_k)\;:=\;\rr{tr}_\Gamma\big(\pi(A_0)\dd\pi(A_1),\ldots,\dd\pi(A_k)\big)
	$$
\end{definition}

As discussed in Remark \ref{rk:quot-cyc}, after the passage to the quotient a quasi-cycle (of dimension $k$) defines a genuine cycle (of dimension $k$). Therefore, one can repeat verbatim the proof of
~\cite[Chapter~3, Section~1.$\alpha$, Proposition~4]{connes-94} or
~\cite[Proposition~8.12]{gracia-varilly-figueroa-01}
to deduce the following result:
\begin{proposition}\label{prop:tau2}
	The character $\tau_k$ of the quasi-cycle $({\Omega}^\bullet,\dd,\rr{tr}_\Gamma)$
	is a cyclic k-cocycle of the algebra $\bb{A}$.
\end{proposition}

In the case of interest of the {magnetic quasi-cycle}
$({\Omega}^\bullet_B, \dd_B, \rr{tr}_\Gamma)$ the associate character is given (up to the multiplicative prefactor $1/2$) by the trilinear map ${\tau}_{B,2}$ defined by~\eqref{eq:ass_charact}. Therefore one has

\begin{corollary}
	The trilinear functional ${\tau}_{B,2}$ defined by \eqref{eq:ass_charact} is a cyclic $2$-cocycle of  the algebra $\bb{S}_B$.
\end{corollary}

Following~\cite[Section~IV.1.$\beta$]{connes-94}
we will refer to the class of  ${\tau}_{B,2}$ in $HC^{\rm even}(\bb{S}_B)$ as the  \emph{Chern character} of the {quasi-even} Fredholm module $(\s{H}_4,F_{B,\varepsilon})$,
and we will denote it with
$Ch_2(\s{H}_4,F_{B,\varepsilon}):=[ {\tau}_{B,2}]\;\in\;HC^{\rm even}(\bb{S}_B)$.

\medskip

Consider the trilinear functional $\rr{Ch}_B$ on $\bb{S}_B$   defined by \eqref{eq:2-cocy-Psi-Dix}. As a consequence of
Lemma \ref{lem:cyc2cocy_Dix} and Lemma \ref{lem:cyc2cocy_Dix-fredh}
one obtains the equality
$$
	\rr{Ch}_B(A_0,A_1,A_2)\;=\;  {\tau}_{B,2}(A_0,A_1,A_2)\;,\qquad\forall \; A_0,A_1,A_2\in \bb{S}_B
$$
The latter equality represents a stronger version of
\cite[Section~IV.2.$\gamma$, Theorem 8]{connes-94} or~\cite[Theorem~10.32]{gracia-varilly-figueroa-01} and provides an incarnation of the celebrated \emph{local index formula} of Connes and Moscovici~\cite{connes-moscovici-95}.

\section{Direct proof of the Second Connes' formula}\label{sec_prop_intro_1}
In this section we will provide the proof of the two versions of the Second Connes' formula anticipated in Section \ref{int_new}.

\medskip

\noindent
{\bf Proof of Lemma \ref{lem:cyc2cocy_Dix}.}
The commutator $[D_B,\pi(A)]$
is well-defined for every $A\in \bb{S}_B$~\cite[Proposition 3.2]{denittis-sandoval-00}, and in view of~\eqref{eq:eve_D} one gets
$$
	\begin{aligned}
		{[D_B,\pi(A)]}\; & =\;[D_{B,-},\pi(A)]                                                                                                     \\&=\;[K_1,A]\;\otimes\;\frac{\gamma_1}{\sqrt{2}}\;+\;[K_2,A]\;\otimes\;\frac{\gamma_2}{\sqrt{2}}\\
		                 & =\;\nabla_1A\;\otimes\;\frac{\ii\gamma_2}{\sqrt{2}\ell_B}\;-\;\nabla_2A\;\otimes\;\frac{\ii\gamma_1}{\sqrt{2}\ell_B}\;,
	\end{aligned}
$$
where $K_1$ and $K_2$ are the {magnetic momenta} \eqref{eq:momenta}, and $\nabla_1$ and $\nabla_2$ are the spatial derivations \eqref{eq:commut_deriv_01_intr}.
Therefore,
$$
	\begin{aligned}
		{[D_B,\pi(A_1)]}{[D_B,\pi(A_2)]} & \;=\;-\frac{1}{2\ell_B^2}\left(
		\nabla_1A_1\nabla_1A_2\otimes\gamma_2^2+\nabla_2A_1\nabla_2A_2\otimes\gamma_1^2
		\right)                                                                    \\
		                                 & \phantom{=}\;+\frac{1}{2\ell_B^2}\left(
		\nabla_1A_1\nabla_2A_2\otimes\gamma_2\gamma_1+\nabla_2A_1\nabla_1A_2\otimes\gamma_1\gamma_2
		\right)\;,
	\end{aligned}
$$
for every $A_1,A_2\in \bb{S}_B$. Let us introduce the notation
\[
	\begin{aligned}
		\delta_0(A_1,A_2)\;: & =\;\nabla_1A_1\nabla_1A_2\;+\;\nabla_2A_1\nabla_2A_2\;, \\
		\delta_1(A_1,A_2)\;: & =\;\nabla_1A_1\nabla_2A_2\;-\;\nabla_2A_1\nabla_1A_2\;.
	\end{aligned}
\]
Since the algebra $\bb{S}_B$ is closed under the action of the derivations, one has that both $\delta_0(A_1,A_2)$ and $\delta_1(A_1,A_2)$ are in $\bb{S}_B$.
By observing that $\gamma_j^2={\bf 1}_4$ is the $4\times 4$ identity matrix for every $j=1,2,3,4$, one gets that
\begin{equation}\label{eq:app_01}
	\begin{aligned}
		{[D_B,\pi(A_1)]}{[D_B,\pi(A_2)]} & \;=\;-\frac{1}{2\ell_B^2}\pi\big(\delta_0(A_1,A_2)
		\big)\;+\;\frac{\ii}{2\ell_B^2}\pi\big(\delta_1(A_1,A_2)
		\big)\;\Gamma
	\end{aligned}
\end{equation}
where $\Gamma$ is the involution defined by \eqref{eq:inv_Gamma}.
As a consequence, it follows that
\[
	\begin{aligned}
		\Gamma\pi(A_0){[D_B,\pi(A_1)]}{[D_B,\pi(A_2)]}\;=\; & -\frac{1}{2\ell_B^2}\pi\big(A_0\delta_0(A_1,A_2)
		\big)\;\Gamma                                                                                              \\
		                                                    & +\;\frac{\ii}{2\ell_B^2}\pi\big(A_0\delta_1(A_1,A_2)
		\big)\;.
	\end{aligned}
\]
In view of \eqref{eq:meas_prop_trip}, one has that
\begin{equation}\label{eq:help_meas}
	|D_{B,\varepsilon}|^{-2}\;\Gamma\pi(A_0){[D_B,\pi(A_1)]}{[D_B,\pi(A_2)]}\;\in\; \rr{S}^{1^+}_{\rm m}\;,
\end{equation}
for every $A_0,A_1,A_2\in\bb{S}_B$.
Therefore, one is allowed to compute the Dixmier trace, and by linearity one gets that
\[
	\begin{aligned}
		\rr{Ch}_B(A_0,A_1,A_2)\;=\;\frac{\ii}{4\ell_B^2}\big(F_1(A_0,A_1,A_2)\;+\:\ii\;F_0(A_0,A_1,A_2)\big)\;,
	\end{aligned}
\]
where
\[
	\begin{aligned}
		F_0(A_0,A_1,A_2)\;: & =\;\Tr_{\rm Dix}\left(|D_{B,\varepsilon}|^{-2}\;\pi\big(A_0\delta_0(A_1,A_2)
		\big)\;\Gamma\right)                                                                               \\
		F_1(A_0,A_1,A_2)\;: & =\;\Tr_{\rm Dix}\left(|D_{B,\varepsilon}|^{-2}\;\pi\big(A_0\delta_1(A_1,A_2)
		\big)\right)\;.
	\end{aligned}
\]
By using~\cite[Lemma B.3]{denittis-gomi-moscolari-19} and the diagonal representation of  $|D_{B,\varepsilon}|^{-2}= (D_B^2+\varepsilon{\bf 1})^{-1}$ in terms of the harmonic oscillator
$Q_{B}$, one gets that
$$
	F_0(A_0,A_1,A_2)\;=\;\Tr_{\rm Dix}\left(\frac{1}{Q_B+\xi{\bf 1}}\; A_0\delta_0(A_1,A_2)
	\right)\; \Tr_{\C^4}(\ii\gamma_1\gamma_2)\;=\;0\;,
$$
since $\Tr_{\C^4}(\ii\gamma_1\gamma_2)=0$. It is  worth remarking that the
equality is justified by the fact that the
first factor of the central term is well-defined for every $\xi>-1$ and its value does not depend on $\xi$~\cite[Proposition 2.27]{denittis-sandoval-00}.
With a similar argument one gets
$$
	F_1(A_0,A_1,A_2)\;=\;4\Tr_{\rm Dix}\left(\frac{1}{Q_B+\xi{\bf 1}}\;\; A_0\delta_1(A_1,A_2)
	\right)
$$
where the pre-factor comes from $\Tr_{\C^4}({\bf 1}_4)=4$.
By putting together all these results, one gets
\[
	\begin{aligned}
		\rr{Ch}_B(A_0,A_1,A_2)\; & =\;\frac{\ii}{\ell_B^2}\; \Tr_{\rm Dix}\left(\frac{1}{Q_B+\xi{\bf 1}}\; A_0\delta_1(A_1,A_2)
		\right)                                                                                                                 \\
		                         & =\;\frac{\ii}{\ell_B^2}\; \fint_B\left( A_0\delta_1(A_1,A_2)
		\right)
		\\
		                         & =\;\frac{\ii}{\ell_B^2}\; \Psi_B\left( A_0,A_1,A_2
		\right)\;,
	\end{aligned}
\]
where the second equality is proved in~\cite[Proposition 2.27]{denittis-sandoval-00} and the third equality follows from~\eqref{eq:2-cocy-Psi}. This concludes the proof.
\qed
\medskip

\begin{remark}\label{rk:triv_inv_2}
	The same proof described above can be adapted to prove the claim of Remark \ref{rk:triv_inv} about the triviality of $\widehat{\rr{Ch}_{B}}$. The main difference relies in the equality
	\[
		\begin{aligned}
			\chi\pi(A_0){[D_B,\pi(A_1)]}{[D_B,\pi(A_2)]}\;=\; & -\frac{1}{2\ell_B^2}\pi\big(A_0\delta_0(A_1,A_2)
			\big)\;\chi                                                                                              \\
			                                                  & +\;\frac{\ii}{2\ell_B^2}\pi\big(A_0\delta_1(A_1,A_2)
			\big)\;\chi\Gamma\;.
		\end{aligned}
	\]
	In the computation of the Dixmier trace both summands produce vanishing terms. The first summand vanishes since $\chi$ is responsible for a term proportional to $\Tr_{\C^4}(\gamma_1\gamma_2\gamma_3\gamma_4)=0$ and the second summand vanishes since $\chi\Gamma$ is responsible for a term proportional to $\Tr_{\C^4}(\gamma_3\gamma_4)=0$.
	\hfill $\blacktriangleleft$
\end{remark}

As a preparation for the proof of Lemma~\ref{lem:cyc2cocy_Dix-fredh} let us anticipate a result which improves~\cite[Lemma 3.14]{denittis-sandoval-00}. For that we need to define the expression
$$
	\widetilde{I_0}(A_1,A_2)\;:=\;|D_{B,\varepsilon}|^{-2}\;[D_B,\pi(A_1)]\;[D_B,\pi(A_2)]\;, \qquad A_1,A_2\in\bb{S}_B\;.
$$
\begin{lemma}\label{lemma_pre_meg_res}
	Let $A_1,A_2\in\bb{S}_B$. Then it holds true that
	$$
		[F_{B,\varepsilon},\pi(A_1)][F_{B,\varepsilon},\pi(A_2)]\;=\;\widetilde{I_0}(A_1,A_2)\;+\; Z(A_1,A_2)
	$$
	with $Z(A_1,A_2)\in \rr{S}^{1}$.
\end{lemma}

\begin{proof}
	The starting point of the proof is decomposition of the product
	$$
		[F_{B,\varepsilon},\pi(A_1)][F_{B,\varepsilon},\pi(A_2)]\;:=\;\sum_{i=0}^3I_i(A_1,A_2)
	$$
	as presented in~\cite[eq. (3.12)]{denittis-sandoval-00}.
	By using the equalities $[F_{B,\varepsilon},T^*]^*=-[F_{B,\varepsilon},T]$ and $[D_{B},T^*]^*=-[D_{B},T]$ for every $T\in\pi(\bb{S}_B)$ one gets
	\begin{equation}\label{eq:fact_comm_F_terms_Is}
		\begin{aligned}
			I_0(A_1,A_2)\; & :=\;|D_{B,\varepsilon}|^{-1}[D_B,\pi(A_1)][D_B,\pi(A_2)]|D_{B,\varepsilon}|^{-1}                          \\
			I_1(A_1,A_2)\; & :=\;|D_{B,\varepsilon}|^{-1}[D_B,\pi(A_1)]D_B\left[|D_{B,\varepsilon}|^{-1},\pi(A_2)\right]               \\
			I_2(A_1,A_2)\; & :=\;\left[|D_{B,\varepsilon}|^{-1},\pi(A_1)\right]D_B[D_B,\pi(A_2)]|D_{B,\varepsilon}|^{-1}               \\
			I_3(A_1,A_2)\; & :=\;\left[|D_{B,\varepsilon}|^{-1},\pi(A_1)\right] D_B^2\left[|D_{B,\varepsilon}|^{-1},\pi(A_2)\right]\;.
		\end{aligned}
	\end{equation}
	Let us observe that a direct computation shows that
	$[|D_{B,\varepsilon}|^{-1},\pi(A)]$ is a diagonal matrix with entries given by $C_{\varepsilon+s,\varepsilon+s}(A)$, $s\in\{0,\pm1\}$, where  the notation of Lemma \ref{lemma:help_ideal_basic} has been used. As a consequence it turns out that $[|D_{B,\varepsilon}|^{-1},\pi(A)]\in \rr{S}^{1}$
	for every $A\in\bb{S}_B$.
	Let us focus on the term $I_1(A_1,A_2)=BS$, where $S:=F_{B,\varepsilon}\left[|D_{B,\varepsilon}|^{-1},\pi(A_2)\right]$ is a trace-class operator and
	$$
		\begin{aligned}
			B\; & :=\;|D_{B,\varepsilon}|^{-1}\;[D_B,\pi(A_1)]\;|D_{B,\varepsilon}|                             \\
			    & =\;-|D_{B,\varepsilon}|^{-1}\;\big[|D_{B,\varepsilon}|,[D_B,\pi(A_1)]\big]\;+\;[D_B,\pi(A_1)]
		\end{aligned}
	$$
	is a bounded operator in view of~\cite[Corollary 3.3]{denittis-sandoval-00} which shows that $[D_B,\pi(A_1)]\in\bb{S}_B\otimes {\rm Mat}_4(\C)$, and~\cite[Proposition  3.4]{denittis-sandoval-00}
	which shows that the commutator $[|D_{B,\varepsilon}|,T]$ is bounded for every $T\in\bb{S}_B\otimes {\rm Mat}_4(\C)$.
	It follows that $I_1(A_1,A_2)\in \rr{S}^{1}$.
	From the identity $I_2(A_1,A_2)=I_2(A_2^*,A_1^*)^*$ one immediately concludes that $I_2(A_1,A_2)\in \rr{S}^{1}$. Similarly,  the term
	$I_3(A_1,A_2)=S'B'$ is the product of the trace-class operator $S':=[|D_{B,\varepsilon}|^{-1},\pi(A_1)]$ and the bounded operator
	$$
		\begin{aligned}
			B'\;: & =\;D_B^2\;\left[|D_{B,\varepsilon}|^{-1},\pi(A_2)\right]                                                       \\
			      & =\; \left[D_B^2|D_{B,\varepsilon}|^{-1},\pi(A_2)\right]\;-\;
			\left[D_B^2, \pi(A_2)\right]\; |D_{B,\varepsilon}|^{-1}                                                                \\
			      & =\;\big[ |D_{B,\varepsilon}| ,\pi(A_2)\big]\;-\;\varepsilon\left[|D_{B,\varepsilon}|^{-1},\pi(A_2)\right]\;-\;
			\left[D_B^2, \pi(A_2)\right]\; |D_{B,\varepsilon}|^{-1}\end{aligned}
	$$
	where in the last equality  it has been used the identity $D_B^2=|D_{B,\varepsilon}| ^2-\varepsilon{\bf 1}$ and the boundedness of
	$[D_B^2, \pi(A_2)]$ is discussed in the proof of~\cite[Lemma 3.9]{denittis-sandoval-00}. Therefore, one obtains that
	$$
		\begin{aligned}
			[F_{B,\varepsilon},\pi(A_1)][F_{B,\varepsilon},\pi(A_2)]\; & =\;{I_0}(A_1,A_2)\;+\; Z'(A_1,A_2)                            \\
			                                                           & =\;\widetilde{I_0}(A_1,A_2)\;+\;Z'(A_1,A_2)\;+\; Z''(A_1,A_2)
		\end{aligned}
	$$
	with $Z'(A_1,A_2)=\sum_{i=1}^3I_i(A_1,A_2)\in \rr{S}^{1}$ and
	$$
		\begin{aligned}
			Z''(A_1,A_2)\;: & =\;{I_0}(A_1,A_2)\;-\;\widetilde{I_0}(A_1,A_2)                                                      \\
			                & =\;-|D_{B,\varepsilon}|^{-1}\;\left[|D_{B,\varepsilon}|^{-1},[D_B,\pi(A_1)][D_B,\pi(A_2)]\right]\;.
		\end{aligned}
	$$
	since $[D_B,\pi(A_1)][D_B,\pi(A_2)]$ is a diagonal matrix with entries in $\bb{S}_B$ in view of \eqref{eq:app_01} one gets that the commutator with $|D_{B,\varepsilon}|^{-1}$ diagonal matrix with entries of the type $C_{\varepsilon ,\varepsilon }(A)$ with  the notation of Lemma \ref{lemma:help_ideal_basic}. As a consequence    $Z''(A_1,A_2)\in \rr{S}^{1}$ and the claim is proved.
\end{proof}

\begin{proof}[Proof of Lemma~\ref{lem:cyc2cocy_Dix-fredh}]
	Since $\rr{S}^{1}\subset \rr{S}^{1^+}_0$, and the ideal $\rr{S}^{1^+}_0$ lies in the common kernel of all the Dixmier traces, one infers from Lemma \ref{lemma_pre_meg_res} that
	$$
		\Tr_{{\rm Dix},\omega}\left(Y \dd_B\pi(A_1)\dd_B\pi(A_2)\right)\;=\;\Tr_{{\rm Dix},\omega}\left(Y \widetilde{I_0}(A_1,A_2)\right)\;.
	$$
	To prove Lemma \ref{lem:cyc2cocy_Dix-fredh} we have to fix $Y=\Gamma\pi(A_0)$ for some  $A_0\in \bb{S}_B$. Since $\Gamma$ commutes with $|D_{B,\varepsilon}|^{-2}$ and
	$$
		\pi(A_0)\;|D_{B,\varepsilon}|^{-2}\;-\;|D_{B,\varepsilon}|^{-2}\; \pi(A_0)\;\in\;\rr{S}^{1^+}_0\;
	$$
	in view of
	Lemma \ref{lemma:help_ideal}, one  finally gets
	$$
		\begin{aligned}
			\Tr_{{\rm Dix},\omega} & \left(\Gamma\pi(A_0) \dd_B\pi(A_1)\dd_B\pi(A_2)\right)\;=                                                   \\
			                       & =\;\Tr_{{\rm Dix},\omega}\left(|D_{B,\varepsilon}|^{-2}\Gamma\pi(A_0)[D_B,\pi(A_1)][D_B,\pi(A_2)]\right)\;.
		\end{aligned}
	$$
	In view of  \eqref{eq:help_meas}  the operator inside the Dixmier trace on the
	right-hand side of the latter equation is a measurable element of the Dixmier
	ideal. As a consequence we do not have to specify a generalized limit for the
	computation of the Dixmier trace. In addition a comparison with \eqref{eq:2-cocy-Psi-Dix} provides
	$$
		\begin{aligned}
			\Tr_{{\rm Dix}} & \left(\Gamma\pi(A_0) \dd_B\pi(A_1)\dd_B\pi(A_2)\right)\;=\;2\;\rr{Ch}_B(A_0,A_1,A_2)\;.
		\end{aligned}
	$$
	By using the notation of the {compatible graded trace} \eqref{eq:dir_trac} and Lemma \ref{lem:cyc2cocy_Dix} one gets
	$$
		\begin{aligned}
			\rr{tr}_\Gamma\left(\pi(A_0) \dd_B\pi(A_1)\dd_B\pi(A_2)\right)\;=\;\frac{\ii 2}{\ell_B^2}\; \Psi_B(A_0,A_1,A_2)\;.
		\end{aligned}
	$$
	A comparison with definition \eqref{eq:ass_charact} concludes the proof.
\end{proof}

\begin{corollary}\label{cor_closeddNes}
	It holds true that
	$$
		\rr{tr}_\Gamma\big([F_{B,\varepsilon},\pi(A_1)][F_{B,\varepsilon},\pi(A_2)]\big)\;=\;0
	$$
	for every   $A_1,A_2\in\bb{S}_B$.
\end{corollary}

\begin{proof}
	From Lemma \ref{lemma_pre_meg_res} and the definition of the {compatible graded trace} \eqref{eq:dir_trac} one obtains that
	$$
		\rr{tr}_\Gamma\big([F_{B,\varepsilon},\pi(A_1)][F_{B,\varepsilon},\pi(A_2)]\big)\;=\; \Tr_{{\rm Dix}}\left(\Gamma \widetilde{I_0}(A_1,A_2)\right)\;.
	$$
	In view of equation \eqref{eq:app_01} one gets
	$$
		\Gamma\widetilde{I_0}(A_1,A_2)\;=\;-\frac{1}{2\ell_B^2}|D_{B,\varepsilon}|^{-2}\pi\big(\delta_0(A_1,A_2)
		\big)\;\Gamma\;+\;\frac{\ii}{2\ell_B^2}|D_{B,\varepsilon}|^{-2}\pi\big(\delta_1(A_1,A_2)\;.
		\big)
	$$
	By using the vanishing of the trace of $\Gamma$ as in the proof of
	Lemma \ref{lem:cyc2cocy_Dix} one obtains
	$$
		\begin{aligned}
			\rr{tr}_\Gamma\big([F_{B,\varepsilon},\pi(A_1)][F_{B,\varepsilon},\pi(A_2)]\big)\; & =\; \frac{\ii}{2\ell_B^2}\Tr_{{\rm Dix}}\left(|D_{B,\varepsilon}|^{-2}\pi\big(\delta_1(A_1,A_2)\right) \\
			                                                                                   & =\; \frac{\ii 2}{\ell_B^2}\Tr_{{\rm Dix}}\left(\frac{1}{Q_B+\xi{\bf 1}}\; \delta_1(A_1,A_2)
			\right)                                                                                                                                                                                     \\
			                                                                                   & =\; \frac{\ii 2}{\ell_B^2}\;\fint_B\big( \delta_1(A_1,A_2)\big)
		\end{aligned}
	$$
	where $\xi>-1$ and the last equality is a consequence of~\cite[Proposition 2.27]{denittis-sandoval-00}. By using the Leibniz's rule for derivations $\nabla_1$ and $\nabla_2$ one obtains that
	$$
		\begin{aligned}
			\nabla_1A_1\nabla_2A_2\; & =\;\nabla_1(A_1\nabla_2A_2)\;-\;A_1(\nabla_1\circ\nabla_2A_2)\;, \\
			\nabla_2A_1\nabla_1A_2\; & =\;\nabla_2(A_1\nabla_1A_2)\;-\;A_1(\nabla_2\circ\nabla_1A_2)\;.
		\end{aligned}
	$$
	Since the derivations $\nabla_1$ and $\nabla_2$ commute one gets
	$$
		\delta_1(A_1,A_2)\;:=\;\nabla_1(A_1\nabla_2A_2)\;-\;\nabla_2(A_1\nabla_1A_2)\;.
	$$
	The property $\fint_B\circ \nabla_j=0$ for $j=1,2$ (\cf~\cite[Section 2.8]{denittis-sandoval-00}) concludes the proof.
\end{proof}

\appendix

\section{Technicalities}\label{app:tech_res}

\subsection{Weak $L^p$-spaces}\label{app_w-lp-dual}
The information contained in this section is quite standard and can be found in
numerous  publications existing in the literature. For the benefit of the reader we will refer mainly to~\cite{pietsch-80,simon-76,simon-05,connes-94,lord-sukochev-zanin-12,azamov-mcdonald-sukochev-zanin-19}.

\medskip

Let $T\in\bb{K}(\s{H})$ be a compact operator.
The Schatten quasi-norm of order $p>0$ is defined as
\begin{equation}
	\|T\|_{p}\;:=\;\left(\sum_{m\in \mathbb{N}_{0}}\mu_m(T)^{p}\right)^{\frac{1}{p}}
\end{equation}
where $\mu_m(T)$ denotes the sequence of singular values of $T$
listed in decreasing order and repeated according to their multiplicity. This is a norm for $p\geqslant 1$. The corresponding
Schatten ideal of order $p$ is defined as
$$
	\rr{S}^{p}\;:=\;\{T\in\bb{K}(\s{H})\;|\;\|T\|_{p} <+\infty \}\;.
$$

\medskip

The \emph{weak} Schatten quasi-norm of order $p>0$ is defined as
\begin{equation}\label{eq:q-norm1}
	\|T\|_{p^+}^\star\;:=\;\sup_{m\geqslant 0}\left\{(m+1)^{\frac{1}{p}}\mu_m(T)\right\}\
\end{equation}
The \emph{weak} Schatten ideal of order $p$ is given by
$$
	\rr{S}^{p}_{\rm w}\;:=\;\{T\in\bb{K}(\s{H})\;|\;\|T\|_{p^+}^\star<+\infty \}\;.
$$
This is a two sided quasi-Banach ideal. The minimal ideal $\rr{S}^{p}_{{\rm w},0}\subset \rr{S}^{p}_{\rm w}$ is defined by
$$
	\rr{S}^{p}_{{\rm w},0}\;:=\;\left\{T\in\rr{S}^{p}_{\rm w}\;\left|\; \lim_{m\to+\infty}\;(m+1)^{\frac{1}{p}}\mu_m(T)=0\right\}\right.\;.
$$
The quasi-norm \eqref{eq:q-norm1} meets the following algebraic properties
\begin{equation}\label{eq:q-norm2}
	\begin{aligned}
		 & \|AT\|_{p^+}^\star\;\leqslant\;\|A\|\; \|T\|_{p^+}^\star\;,\quad  \|TA\|_{p^+}^\star\;\leqslant\;\|A\|\; \|T\|_{p^+}^\star\;, \\
		 & \|T_1+T_2\|_{p^+}^\star\;\leqslant\;2^{\frac{1}{p}}\left(\|T_1\|_{p^+}^\star+\|T_2\|_{p^+}^\star\right)\;,
	\end{aligned}
\end{equation}
valid for every $T,T_1,T_2\in \rr{S}^{p}_{\rm w}$ and $A\in\bb{B}(\s{H})$. Moreover, one has a H\"older type inequality,
\begin{equation}\label{eq:q-norm3}
	\begin{aligned}
		&\|T_1T_2\|_{r^+}^\star\;\leqslant\;2^{\frac{1}{r}}\;\|T_1\|_{p^+}^\star\; \|T_2\|_{q^+}^\star\end{aligned}
\end{equation}
valid for every $T_1\in \rr{S}^{p}_{\rm w}$ and $T_2\in \rr{S}^{q}_{\rm w}$ such that $r^{-1}=p^{-1}+q^{-1}$.
A proof of \eqref{eq:q-norm2} and \eqref{eq:q-norm3} can be found in~\cite[Theorem 2.1]{simon-76}.  For $p>q$, one can  prove that~\cite[eq. (2.2)]{azamov-mcdonald-sukochev-zanin-19}
$$
	\|T\|_p\;\leqslant\;\rr{Z}\left(\frac{p}{q}\right)^{\frac{1}{p}}\; \|T\|_{q^+}^\star
$$
where $\|\;\|_p$ is the norm of the Schatten ideal $\rr{S}^{p}$.
As a consequence one obtains the continuous embeddings $\rr{S}^{q}_{\rm w}\subset \rr{S}^{p}$ if $p>q$. Since the condition $T\in \rr{S}^{p}$ implies that $\mu_m(T)^p$ must  decrease faster than $m^{-1}$ one obtains the chain of  continuous embeddings
$$
	\rr{S}^{q}_{\rm w}\;\subset\; \rr{S}^{p}\;\subset\;\rr{S}^{p}_{{\rm w},0}\; \;\subset\;\rr{S}^{p}_{{\rm w}}\;,\qquad 0<q<p<+\infty\;.
$$

\medskip

For $p>1$ the weak Schatten quasi-norm is equivalent to
Calder\'on
norm
$$
	\|T\|_{p^+}\;:=\;\sup_{N\geqslant 1}\left\{\frac{1}{N^{1-\frac{1}{p}}}\sum_{m=0}^{N-1}\mu_m(T)\right\}\;
$$
since one has the inequalities
$$
	\|T\|_{p^+}^\star\;\leqslant\;\|T\|_{p^+}\;\leqslant\;\frac{p}{p-1}\|T\|_{p^+}^\star\;.
$$
Therefore, when $p>1$ the spaces $\rr{S}^{p}_{\rm w}$ are indeed two-sided Banach ideals. In particular one has that $\rr{S}^{p}_{\rm w}=\rr{S}^{p^+}$ coincides with the $p$-th Dixmier ideal   as defined in~\cite[Appendix B.1]{denittis-gomi-moscolari-19}.

\medskip

For $p=1$  the appropriate Calder\'on
norm is
$$
	\|T\|_{1^+}\;:=\;\sup_{N\geqslant 2}\left\{\frac{1}{\log(N)}\sum_{m=0}^{N-1}\mu_m(T)\right\}
$$
and the corresponding Dixmier ideal is
$$
	\rr{S}^{1^+}\;:=\;\{T\in\bb{K}(\s{H})\;|\;\|T\|_{1^+}<+\infty \}\;.
$$
Also in this case
$\rr{S}^{1^+}_0$ denotes the closure of the finite-rank operators in the  Calder\'on
norm  $\|\;\|_{1^+}$. In this case the weak Schatten ideal $\rr{S}^{1}_{\rm w}$ is just smaller than $\rr{S}^{1^+}$. In fact one has the chain of strict inclusions
$$
	\rr{S}^{1}\;\subset\; \rr{S}^{1^+}_0\;\subset\;\rr{S}^{1}_{{\rm w}}\; \;\subset\;\rr{S}^{1^+}\;.
$$

\medskip

For $q\geqslant 1$, let us introduce the Ma\u{c}aev norm
$$
	\|T\|_{q^-}\;:=\;\sum_{m=0}^{+\infty}\frac{\mu_m(T)}{(m+1)^{1-\frac{1}{q}}}\;,
$$
and the Ma\u{c}aev ideal of order $q$
$$
	\rr{S}^{q^-}\;:=\;\{T\in\bb{K}(\s{H})\;|\;\|T\|_{q^-}<+\infty \}\;.
$$
one has the following strict inclusions
$$
	\rr{S}^{s^+}\;\subset\;\rr{S}^{q^-}\;\subset\; \rr{S}^{q}\;\subset\;\rr{S}^{q^+}\;,\qquad 1\leqslant s <q<+\infty\;.
$$
The space $\rr{S}^{p^+}$ is the dual of $\rr{S}^{q^-}$ when $q^{-1}+p^{-1}=1$, $p,q\geqslant 1$. This means that if $T\in \rr{S}^{p^+}$ and $S\in\rr{S}^{q^-}$, then $TS\in\rr{S}^{1}$
and $
	\|TS\|_1\;\leqslant\;\|T\|_{p^+}\; \|S\|_{q^-}$\;.

\subsection{Relevant trace-class elements}
Let $Q_B$ be the {harmonic oscillator} \eqref{eq:harm_osc}, and let us define
$$
	Q_{B,\varepsilon}\;:=\;Q_B\;+\;\varepsilon{\bf 1}\;,\qquad \varepsilon>-1\;.
$$
\begin{lemma}\label{lemma:help_ideal_basic}
	It holds true that
	\begin{equation}
		\label{eq:lemma-help_ideal_basic}
		C_{\varepsilon,\varepsilon'}(A)\;:=\;Q_{B,\varepsilon}^{-\frac{1}{2}}\;A\;-\;A\;Q_{B,\varepsilon'}^{-\frac{1}{2}}\;\in\;\rr{S}^{1}
	\end{equation}
	for every $A\in\bb{S}_B$, independently of $\varepsilon,\varepsilon'>-1$.
\end{lemma}

\begin{proof}
	Let us start with the simple case of $A=\Upsilon_{j\mapsto k}$, where the   {transition operator} is defined by \eqref{eq:intro:basic_op}. In the proof of~\cite[Lemma 3.14]{denittis-sandoval-00} it has been proved that
	$$
		C_{\varepsilon,\varepsilon'}(\Upsilon_{j\mapsto k})\;:=\;Q_{B,\varepsilon}^{-\frac{1}{2}}\Upsilon_{j\mapsto k}\;-\;\Upsilon_{j\mapsto k}Q_{B,\varepsilon'}^{-\frac{1}{2}}\;\in\;\rr{S}^{1}\subset\rr{S}^{1^+}_0
	$$
	is a trace class element. This follows by observing that the    singular values are given by
	$$
		\mu_m\left[C_{\varepsilon,\varepsilon'}(\Upsilon_{j\mapsto k})\right]\;:=\;\frac{\alpha_m^{(j,k)}}{(m+1)^{\frac{3}{2}}}\;\propto\;(m+1)^{-\frac{3}{2}}
	$$
	with
	$$
		\alpha_m^{(j,k)}\;:=\;\frac{|\zeta_j-\zeta_k|}{\sqrt{1+\frac{\zeta_j}{m+1}}\sqrt{1+\frac{\zeta_k}{m+1}}\left(\sqrt{1+\frac{\zeta_j}{m+1}}+\sqrt{1+\frac{\zeta_k}{m+1}}\right)}\;\leqslant\;\frac{|\zeta_j-\zeta_k|}{2}
	$$
	where $\zeta_j:=j+ \varepsilon'$, $\zeta_k:=k+\varepsilon$. Moreover, from the explicit form of the singular values one infers that
	$$
		\left\|C_{\varepsilon,\varepsilon'}(\Upsilon_{j\mapsto k})\right\|_{1}\;\leqslant\;\frac{1}{2}\;\rr{Z}\left(\frac{3}{2}\right)\;|\zeta_j-\zeta_k|
	$$
	where $\rr{Z}$ denotes the Riemann zeta function.
	Now, let us consider a generic element
	$A=\sum_{(j,k)\in\N_0^2}a_{j,k}\Upsilon_{j\mapsto k}$ with $\{a_{j,k}\}\in S(\N_0^2)$. The linearity of the commutator,  the triangular inequality and  the  Cauchy–Schwarz inequality imply
	$$
		\left\|Q_{B,\varepsilon}^{-\frac{1}{2}}A\;-\;AQ_{B,\varepsilon}^{-\frac{1}{2}}\right\|_{1}\;\leqslant\;
		\frac{1}{2}\;\rr{Z}\left(\frac{3}{2}\right)\;\sum_{(j,k)\in\N_0^2}|\zeta_j-\zeta_k|\;|a_{j,k}|\;\leqslant\;M_{\varepsilon,\varepsilon'}^p\;r_p(\{a_{j,k}\})
	$$
	where
	$$
		M_{\varepsilon,\varepsilon'}^p\;:=\;\frac{1}{2}\;\rr{Z}\left(\frac{3}{2}\right)\;\left(\sum_{(j,k)\in\N_0^2}\frac{|\zeta_j-\zeta_k|}{(2j+1)^p(2k+1)^p}\right)^{\frac{1}{2}}
	$$
	is a finite constant whenever $p>2$ and
	$$
		r_p(\{a_{j,k}\})\;:=\;\left(\sum_{(j,k)\in\N_0^2}{(2j+1)^p(2k+1)^p|a_{j,k}|^2}\right)^{\frac{1}{2}}
	$$
	is the $p$-th Schwarz semi-norm of the Fr\'echet space $S(\N_0^2)$.
	This proves that $C_{\varepsilon,\varepsilon'}:\bb{S}_B\to \rr{S}^{1}$ is a continuous map.
\end{proof}

The next result can be proved along the same lines as those of the proof of Lemma  \ref{lemma:help_ideal_basic}. However, it provides two  vanishing criteria for elements in   $\bb{L}^1_B$ and not only in $\bb{S}_B$. The proof of the next result justifies the formulas
anticipated after~\cite[Proposition 2.27]{denittis-sandoval-00} .

\begin{lemma}\label{lemma:help_ideal}
	It holds true that
	\begin{equation}
		\label{eq:appendix-commutator-with-Q-trace-class-I}
		D_{\varepsilon,\varepsilon'}(A)\;:=\;Q_{B,\varepsilon}^{-1}\;A\;-\;A\;Q_{B,\varepsilon'}^{-1}\;\in\;\rr{S}^{1^+}_0\;,
	\end{equation}
	and
	\begin{equation}
		\label{eq:appendix-commutator-with-Q-trace-class-I}
		J_{\varepsilon,\varepsilon',\varepsilon''}(A)\;:=\;Q_{B,\varepsilon}^{-\frac{1}{2}}\;A\; Q_{B,\varepsilon'}^{-\frac{1}{2}}\;-\;Q_{B,\varepsilon''}^{-1}\;A\;\in\;\rr{S}^{1^+}_0\;,
	\end{equation}
	for every $A\in\bb{L}^1_B$, independently of $\varepsilon,\varepsilon',\varepsilon''>-1$.
\end{lemma}

\begin{proof}
	Let us start with the simple case of $A=\Upsilon_{j\mapsto k}$. By using the same notation in~\cite[Corollary 2.26]{denittis-sandoval-00} one gets the spectral resolution
	$$
		D_{\varepsilon,\varepsilon'}(\Upsilon_{j\mapsto k})\;=\;\left(\sum_{m\in\N_0}\frac{\zeta_j-\zeta_k}{(m+1+\zeta_j)(m+1+\zeta_k)}P_m\right)\Upsilon_{j\mapsto k}\;,
	$$
	where $\zeta_j:=j+\varepsilon'$, $\zeta_k:=k+\varepsilon$
	and the $P_m$ are the dual Landau projections defined in~\cite[eq. (2.15)]{denittis-sandoval-00}. It follows that
	$$
		\big|D_{\varepsilon,\varepsilon'}(\Upsilon_{j\mapsto k})\big|\;=\;\left(\sum_{m\in\N_0}\frac{|\zeta_j-\zeta_k|}{(m+1+\zeta_j)(m+1+\zeta_k)}P_m\right)\Pi_j
	$$
	which shows that the sequence of singular values of $D_{\varepsilon,\varepsilon'}(\Upsilon_{j\mapsto k})$ behaves asymptotically as
	\begin{equation}\label{eq:sing_va-D}
		\mu_m\left[D_{\varepsilon,\varepsilon'}(\Upsilon_{j\mapsto k})\right]\;:=\;\frac{|\zeta_j-\zeta_k|}{(m+1+\zeta_j)(m+1+\zeta_k)}\;\propto\;(m+1)^{-2}
	\end{equation}
	and has multiplicity $1$. As a consequence one gets that
	$D_{\varepsilon,\varepsilon'}(\Upsilon_{j\mapsto k})$ is trace-class and in turn
	$D_{\varepsilon,\varepsilon'}(\Upsilon_{j\mapsto k})\in\rr{S}^{1}\subset\rr{S}^{1^+}_0$. Moreover, from the proof of~\cite[Proposition 2.27]{denittis-sandoval-00} one obtains
	$$
		\|D_{\varepsilon,\varepsilon'}(\Upsilon_{j\mapsto k})\|_{1+}\;\leqslant\;\|Q_{B,\varepsilon}^{-1} \Upsilon_{j\mapsto k}\|_{1+}\;+\; \|\Upsilon_{j\mapsto k}\;Q_{B,\varepsilon}^{-1}\|_{1+}\;\leqslant\; 2\;.
	$$
	Now, let us consider a generic element
	$A=\sum_{(j,k)\in\N_0^2}a_{j,k}\Upsilon_{j\mapsto k}$ with $\{a_{j,k}\}\in\ell^1(\N_0^2)$. The linearity of $C_{\varepsilon,\varepsilon'}$ and the triangular inequality imply
	$$
		\|D_{\varepsilon,\varepsilon'}(A)\|_{1+}\;\leqslant\;
		\sum_{(j,k)\in\N_0^2}|a_{j,k}|\; \|D_{\varepsilon,\varepsilon'}(\Upsilon_{j\mapsto k})\|_{1+}\;\leqslant\;2\; \|\{a_{j,k}\}\|_{\ell^1}\;.
	$$
	This proves that $D_{\varepsilon,\varepsilon'}:\bb{L}^1_B\to \rr{S}^{1^+}$ is a continuous map. Since
	$D_{\varepsilon,\varepsilon'}(A)\in \rr{S}^{1^+}_0$ if $A$ is a finite linear combination of the operators
	$\Upsilon_{j\mapsto k}$, and $\rr{S}^{1^+}_0$ is closed in the norm $\|\;\|_{_{1+}}$, one
	infers that $D_{\varepsilon,\varepsilon'}(\bb{L}^1_B) \subseteq \rr{S}^{1^+}_0$ by continuity.
	The proof for the map $J_{\varepsilon,\varepsilon',\varepsilon''}$ is similar. From  the spectral resolution
	$$
		J_{\varepsilon,\varepsilon',\varepsilon''}(\Upsilon_{j\mapsto k})\;=\;\left(\sum_{m\in\N_0}\frac{m+1+\xi_k-\sqrt{(m+1+\zeta_j)(m+1+\zeta_k)}}{(m+1+\xi_k)\sqrt{(m+1+\zeta_j)(m+1+\zeta_k)}}P_m\right)\Upsilon_{j\mapsto k}\;,
	$$
	with $\xi_k:=k+\varepsilon''$, one gets
	$$
		\mu_m\left[J_{\varepsilon,\varepsilon',\varepsilon''}(\Upsilon_{j\mapsto k})\right]\;=\;\frac{\left|m+1+\xi_k-\sqrt{(m+1+\zeta_j)(m+1+\zeta_k)}\right|}{(m+1+\xi_k)\sqrt{(m+1+\zeta_j)(m+1+\zeta_k)}}\;\propto\;(m+1)^{-2}\;,
	$$
	which implies $J_{\varepsilon,\varepsilon',\varepsilon''}(\Upsilon_{j\mapsto k})\in\rr{S}^{1}\subset\rr{S}^{1^+}_0$.
	Again~\cite[Proposition 2.27]{denittis-sandoval-00} provides the estimate $\|J_{\varepsilon,\varepsilon',\varepsilon''}(\Upsilon_{j\mapsto k})\|_{1+}\leqslant 2$. At this point, the continuity argument follows as in the previous case.
\end{proof}

\begin{remark}\label{rk:trac-class-strong}
	From the explicit form of the singular values of $D_{\varepsilon,\varepsilon'}(\Upsilon_{j\mapsto k})$ provided in \eqref{eq:sing_va-D} one infers that
	$$
		\|D_{\varepsilon,\varepsilon'}(\Upsilon_{j\mapsto k})\|_{1}\;\leqslant\;\rr{Z}\left(3\right)\;|\zeta_j-\zeta_k|
		\;
	$$
	where $\rr{Z}$ denotes the Riemann zeta function. Therefore, one can use the same argument in the proof of Lemma \ref {lemma:help_ideal_basic}  to show that $D_{\varepsilon,\varepsilon'}:\bb{S}_B\to \rr{S}^{1}$ is a continuous map. A similar result also holds for $J_{\varepsilon,\varepsilon',\varepsilon''}$.
	\hfill $\blacktriangleleft$
\end{remark}

For the next result we need to introduce the creation and annihilation operators
\begin{equation}\label{eq:b-oper}
	\rr{b}^\pm\;:=\;-\frac{1}{\sqrt{2}}\left(G_1\pm\ii G_2\right)
\end{equation}
and the number operator $N_{\rr{b}}:= \rr{b}^+\rr{b}^-$. From the canonical commutation relation one gets $\rr{b}^-\rr{b}^+=N_{\rr{b}}+1$. If $P_m$ is the dual Landau projection, then one has that $N_{\rr{b}}P_m=P_m N_{\rr{b}}=mP_m$.

\begin{lemma}\label{lemma:help_ideal_basic-+b}
	It holds true that
	\begin{equation}
		\label{eq:lemma-commutator-with-Q-and-b}
		\rr{b}^\pm C_{\varepsilon,\varepsilon'}(A)\;:=\;\rr{b}^\pm\left(Q_{B,\varepsilon}^{-\frac{1}{2}}\;A\;-\;A\;Q_{B,\varepsilon'}^{-\frac{1}{2}}\right)\;\in\;\rr{S}^{q^-}\;,\qquad q>1
	\end{equation}
	for every $A\in\bb{S}_B$, independently of $\varepsilon,\varepsilon'>-1$.
\end{lemma}

\begin{proof}
	The proof follows the same strategy of Lemma \ref{lemma:help_ideal_basic}. Using the spectral representation of
	of $C_{\varepsilon,\varepsilon'}(\Upsilon_{j\mapsto k})$ given in~\cite[Lemma 3.14]{denittis-sandoval-00}
	one can compute explicitly the singular values of $\rr{b}^\pm C_{\varepsilon,\varepsilon'}(\Upsilon_{j\mapsto k})$. One obtains  that
	$$
		\mu_m\left[\rr{b}^\pm C_{\varepsilon,\varepsilon'}(\Upsilon_{j\mapsto k})\right]\;=\; \left(m+\frac{1\pm 1}{2}\right)^{\frac{1}{2}}\mu_m\left[C_{\varepsilon,\varepsilon'}(\Upsilon_{j\mapsto k})\right]\;\propto\;(m+1)^{-1}\;
	$$
	where the singular values $\mu_m\left[C_{\varepsilon,\varepsilon'}(\Upsilon_{j\mapsto k})\right]$ are explicitly described in the proof of Lemma \ref{lemma:help_ideal_basic}. By using the definition of the norm $\|\;\|_{q^-}$ given in Appendix \ref{app_w-lp-dual} one obtains that
	$$
		\left\|\rr{b}^\pm C_{\varepsilon,\varepsilon'}(\Upsilon_{j\mapsto k})\right\|_{q^-}\;\leqslant\;\frac{1}{2}\;\rr{Z}\left(2-\frac{1}{q}\right)\;|\zeta_j-\zeta_k|
	$$
	where $\rr{Z}$ is the Riemann zeta function. Therefore, one has that
	$\rr{b}^\pm C_{\varepsilon,\varepsilon'}(\Upsilon_{j\mapsto k})\in \rr{S}^{q^-}$ whenever $q>1$. The same argument used in the proof of Lemma \ref{lemma:help_ideal_basic} provides the continuity of the maps $\rr{b}^\pm C_{\varepsilon,\varepsilon'}:\bb{S}_B\to\rr{S}^{q^-}$ and this concludes the proof.
\end{proof}

\begin{corollary}\label{cor_dual_opo}
	It holds true that
	$$
		T\;\rr{b}^\pm C_{\varepsilon,\varepsilon'}(A)\;\in\;\rr{S}^{1}\;,\quad \qquad \rr{b}^\pm C_{\varepsilon,\varepsilon'}(A)\;T\;\in\;\rr{S}^{1}
	$$
	for every $A\in\bb{S}_B$ and $T\in\rr{S}^{p^+}$, with $p>1$, independently of $\varepsilon,\varepsilon'>-1$.
\end{corollary}

\begin{proof}
	Let $p>1$ and define $q:=p/(p-1)>1$. Since $\rr{S}^{p^+}$ is the dual of $\rr{S}^{q^-}$ (see Appendix \ref{app_w-lp-dual}) and
	using the result proved in Lemma \ref{lemma:help_ideal_basic-+b} one obtains $T\rr{b}^\pm C_{\varepsilon,\varepsilon'}(A)\in \rr{S}^{1}$. The second implication follows from the identity $(A^*B^*)^*=BA$ and the fact that $\rr{S}^{p^\pm}$ and $\rr{S}^{1}$ are self-adjoint ideals.
\end{proof}

\subsection{Quasi-symmetry of the Dirac operator and its consequences}\label{sec:q-sim-dir}
Let us  represent the Dirac operator \eqref{eq:intro_D} as the sum
$$
	D_B\;=\;D_{B,-}\;+\;D_{B,+}
$$
of the two terms
\[
	\begin{aligned}
		D_{B,-}\; & :=\;\frac{1}{\sqrt{2}}\big(K_1\;\otimes\;\gamma_1\;+\;K_2\;\otimes\;\gamma_2\big)\;, \\
		D_{B,+}\; & :=\;\frac{1}{\sqrt{2}}\big(G_1\;\otimes\;\gamma_3\;+\;G_2\;\otimes\;\gamma_4\big)\;.
	\end{aligned}
\]
A simple calculation involving the commutation relations  between the operators
$K_j$, $G_j$, and the matrices $\gamma_j$ provides
\[
	D_{B,+}D_{B,-}\;=\;-D_{B,-}D_{B,+}\;.
\]
The latter equation immediately implies the two relations
\[
	D_B^2\;=\;D_{B,-}^2\;+\;D_{B,+}^2\;,\qquad [D_B^2,D_{B,\pm}]\;=\;0\;.
\]
Finally, a straightforward computation provides
$$
	\Gamma \;D_{B,\pm}\; \Gamma\; =\; \pm D_{B,\pm}\;.
$$
It is worth to point out that all the equations presented above are initially well-defined on the common core $S(\R^2)\otimes\C^4$ of
$D_{B,\pm}$ and then are extended by continuity to the whole Hilbert space $\s{H}_4$. The commutator of $D_{B}$ with elements in
$\pi(\bb{S}_B)$ is well-defined~\cite[Proposition 3.2]{denittis-sandoval-00}.
Since $D_{B,+}$ commutes with $\pi(\bb{S}_B)$ in view of~\cite[Lemma 2.19]{denittis-sandoval-00} it follows that
\begin{equation}\label{eq:eve_D}
	[D_B,\pi(A)]\;=\;[D_{B,-},\pi(A)]
\end{equation}
and in turn
$$
	\Gamma\;[D_B,\pi(A)]\;\Gamma\;=\;-[D_B,\pi(A)]\;.
$$
The latter equation shows that $[D_B,\pi(A)]$ has degree 1 with respect to $\Gamma$.

\medskip

In~\cite[Lemma 3.10]{denittis-sandoval-00} it has been proved that
$[F_{B,\varepsilon},\pi(A)]\in \rr{S}^{2^+}$ if $A\in\bb{S}_B$. By replacing $F_{B,\varepsilon}$ with the anticommutator
$$
	\{\Gamma, F_{B,\varepsilon}\}\;:=\;\Gamma F_{B,\varepsilon}+ F_{B,\varepsilon}\Gamma\;=\;2\Gamma\frac{D_{B,+}}{|D_{B,\varepsilon}|}
$$
one obtains a stronger result.
\begin{lemma}\label{lemma:help_ideal_basic}
	It holds true that
	\begin{equation}
		\label{eq:lemma-help_ideal_basic}
		\big[\{\Gamma, F_{B,\varepsilon}\},\pi(A)\big]\;\in\;\rr{S}^{q^-}\;,\qquad q>1
	\end{equation}
	for every $A\in\bb{S}_B$, independently of $\varepsilon>0$.
\end{lemma}

\begin{proof}
	A direct computation shows that
	$$
		\begin{aligned}
			\{\Gamma, F_{B,\varepsilon}\} & =\;2\Gamma\;\left(\begin{array}{c c c c}
					0        & \rr{b}^+ & 0         & 0         \\
					\rr{b}^- & 0        & 0         & 0         \\
					0        & 0        & 0         & -\rr{b}^- \\
					0        & 0        & -\rr{b}^+ & 0
				\end{array}
			\right)|D_{B,\varepsilon}|^{-1}                                               \\
			                              & =\;2\left(\begin{array}{c c c c}
					0                                              & \rr{b}^+\; Q_{B,\varepsilon}^{-\frac{1}{2}} & 0                                             & 0                                           \\
					-\rr{b}^-\; Q_{B,\varepsilon-1}^{-\frac{1}{2}} & 0                                           & 0                                             & 0                                           \\
					0                                              & 0                                           & 0                                             & \rr{b}^-\; Q_{B,\varepsilon}^{-\frac{1}{2}} \\
					0                                              & 0                                           & \rr{b}^+\; Q_{B,\varepsilon+1}^{-\frac{1}{2}} & 0
				\end{array}
			\right)
		\end{aligned}
	$$
	where the $\rr{b}^\pm$ are the  creation and annihilation operators
	defined by \eqref{eq:b-oper}. Since the operators $\rr{b}^\pm$ commute with $\pi(A)$ one gets that the non-zero elements of the commutator $[\{\Gamma, F_{B,\varepsilon}\},\pi(A)]$ are of the type
	$\rr{b}^\pm C_{\varepsilon,\varepsilon}(A)$. Therefore, the result follows from Lemma \ref{lemma:help_ideal_basic-+b}.
\end{proof}

\section{The cyclic cohomology of the magnetic algebra}\label{app:cyc_cohom}
Cyclic cohomology provides a natural analog of the classical \emph{de Rham
	theory} in the context of noncommutative $C^*$-algebras. A complete description of this theory is presented
in~\cite[Chapter~3]{connes-94} and~\cite[Chapters~8 \& 10]{gracia-varilly-figueroa-01}.
In this section, we will review only the most basic aspects of the
theory.

\medskip

To describe the cyclic cohomology of the magnetic algebra we
start by considering the set of $(n + 1)$-multilinear functionals $\varphi$ defined on $\mathscr{S}_{B}$, satisfying the cyclic condition
\[
	\varphi(A_1,\dots,A_n,A_0)={ (-1) }^n\varphi(A_0,A_1,\dots,A_n),\qquad A_{i} \in \mathscr{S}_{B}\;,
\]
Let $C^n_\lambda(\mathscr{S}_{B})$, with $n \in \mathbb{N}_{0}$, be the linear space of cyclic densely defined
$(n+1)$-linear functionals. The elements of $C^n_\lambda(\mathscr{S}_{B})$ are called
\emph{cyclic $n$-cochains}. On the family of sets $C^n_\lambda(\mathscr{S}_{B})$ acts
the \emph{Hochschild coboundary map}
$b:C^n_\lambda(\mathscr{S}_{B})\to C^{n+1}_\lambda(\mathscr{S}_{B})$ given by
\begin{align*}
	(b\varphi)(A_0,\ldots,A_{n+1}) & \;=\;\sum_{j=0}^n{ (-1) }^n\varphi(A_0,\ldots,A_j A_{j+1},\ldots,A_{n+1}) \\
	                               & \qquad +{ (-1) }^{n+1}\varphi(A_{n+1}A_0,\ldots,A_{n+1})\;.
\end{align*}
From the definition, one gets $b^2=0$. An element $\varphi\in C^n_\lambda(\mathscr{S}_{B})$ is
called \emph{cyclic $n$-cocycle}, if and only if, $b\varphi=0$. Elements of the form
$b\varphi\in C^n_\lambda(\mathscr{S}_{B})$ are called \emph{cyclic $n$-coboundaries}. The
\emph{cyclic cohomology} of $\mathscr{S}_{B}$ is the cohomology of the cyclic
complex $(C^\bullet_\lambda(\mathscr{S}_{B}),b)$, and it is denoted by
$\HC^\bullet(\mathscr{S}_{B})$. More precisely one has that
\[
	\HC^n(\mathscr{S}_{B})\:=\:\frac{{\rm Ker}\left(b:C^n_\lambda(\mathscr{S}_{B})\to C_\lambda^{n+1}(\mathscr{S}_{B})\right)}{{\rm Ran}\big(b:C_\lambda^{n-1}(\mathscr{S}_{B})\to C_\lambda^{n}(\mathscr{S}_{B})\big)}\;,\qquad n \in \mathbb{N}_{0}\;.
\]
Any element of $\HC^n(\mathscr{S}_{B})$ is an equivalence class of {cyclic
		$n$-cocycles} modulo {cyclic $n$-coboundaries}.

\medskip

Let us recall that there exist the \emph{periodicity operator} $S$ which provides group homomorphisms
$S:\HC^n(\mathscr{S}_{B})\to \HC^{n+2}(\mathscr{S}_{B})$
\cite[Section 10.1]{gracia-varilly-figueroa-01}. Using
this operator, one obtains two groups as the inductive limits
$$
	\HC^{\rm even}(\bb{S}_B)\;:=\;\varinjlim\HC^{2n}(\mathscr{S}_{B})\;,\qquad \HC^{\rm odd}(\bb{S}_B)\;:=\;\varinjlim\HC^{2n+1}(\mathscr{S}_{B})\;.
$$
which define the \emph{periodic cyclic cohomology} of $\bb{S}_B$.
The next result is essentially proved in
\cite{elliot-natsume-nest-88}.
\begin{lemma}\label{lemma:comp_cohom}
	It holds true that
	$$
		\HC^{\rm even}(\bb{S}_B)\;=\;\Z\;{\textstyle [\fint_B]}\;,\qquad \HC^{\rm odd}(\bb{S}_B)\;=\;0\;.
	$$
\end{lemma}
\begin{proof}
	In \cite[Theorem 2]{elliot-natsume-nest-88} it is proved that
	$\HC^{\star}(\bb{K}^\infty)\;=\;\HC^{\star}(\C)$ where $\star$  stays for even or odd and $\bb{K}^\infty$ denotes the $\ast$-algebra
	of those
	Hilbert-Schmidt operators on $L^2(\R)$ whose integral kernels belong to $S(\R^2)$.
	By adapting \cite[Theorem 1.30]{folland-89} one obtains that the
	\virg{Weyl transform} $\rho$ provides a $\ast$-isomorphism $\bb{S}_B\simeq\bb{K}^\infty$. As a consequence one has that the
	periodic cyclic cohomology of $\bb{S}_B$ coincides with that of $\C$ which is known to be $\HC^{\rm even}(\C)\simeq\Z$ and
	$\HC^{\rm odd}(\C)=0$. To conclude the proof it is enough to observe that a cyclic 0-cocycle is clearly the same thing as a trace and $\bb{S}_B$ is endowed with the (faithful)
	trace $\fint_B$.
\end{proof}

There are, in principle, two
canonical pairings between periodic cyclic cohomology and $K$-theory~\cite[Section~3.III]{connes-94}. In the specific case of the $\ast$-algebra $\bb{S}_B$ the only relevant pairing is
\[
	\langle\;,\;\rangle:HC^{\rm even}(\bb{S}_B)\times K_0(\bb{S}_B)\to\C\;
\]
implemented by
\[
	\langle[\varphi],[P]\rangle\;:=\;\frac{1}{m!}(\varphi\;\sharp\; {\rm Tr}_{\C^N})(P,\ldots,P)\;,
\]
where $\varphi\in C^{2m}_\lambda(\mathscr{S}_{B})$ is a representative of $[\varphi]\in HC^{\rm even}(\bb{S}_B)$ and the projection
$P\in \mathscr{S}_{B}\otimes {\rm Mat}_N(\C)$ is a representative of $[P]\in K_0(\bb{S}_B)$.
The odd pairing is trivial in view of the fact that
$HC^{\rm odd}(\bb{S}_B)=0= K_1(\bb{S}_B)$.

\end{document}